\documentclass[reqno,11pt]{amsart}
\usepackage{a4wide}
\usepackage[utf8]{inputenc}
\usepackage[T1]{fontenc}
\usepackage{ae,aecompl}
\usepackage[english]{babel}
\usepackage{enumerate}
\usepackage{latexsym}
\usepackage{amsmath,amsthm,amsfonts}
\usepackage{xr}
\usepackage{graphicx}
\usepackage{stmaryrd}
\usepackage{mathtools}
\usepackage{braket}
\usepackage{pgfplots}
\usepackage{textcomp}
\usepackage{bbm}
\usepackage{here}
\usepackage{relsize}
\usepackage{dsfont}

\newtheorem{satz}{Theorem}[section]
\newtheorem{cor}[satz]{Corollary}
\newtheorem{lem}[satz]{Lemma}
\newtheorem{defi}[satz]{Definition}
\newtheorem{bem}[satz]{Remark}
\newtheorem{bsp}[satz]{Example}
\numberwithin{equation}{section}

\newcommand{\erf}{\operatorname{erf}}
\newcommand{\erfc}{\operatorname{erfc}}
\renewcommand{\Re}{\operatorname{Re}}
\renewcommand{\Im}{\operatorname{Im}}
\newcommand{\Arg}{\operatorname{Arg}}
\newcommand{\ran}{\operatorname{ran}}
\newcommand{\sinc}{\operatorname{sinc}}

\title[Integral representation of superoscillations via complex Borel measures]{\bf Integral representation of superoscillations \\ via complex Borel measures and their convergence}

\author[Jussi Behrndt]{Jussi Behrndt}
\address{(JB) Institute of Applied Mathematics, Graz University of Technology, Steyrergasse 30, 8010 Graz, Austria}
\email{behrndt@tugraz.at}

\author[Fabrizio Colombo]{Fabrizio Colombo}
\address{(FC) Politecnico di Milano, Dipartimento di Matematica, Via E. Bonardi, 9, 20133 Milano, Italy}
\email{fabrizio.colombo@polimi.it}

\author[Peter Schlosser]{Peter Schlosser}
\address{(PS) Politecnico di Milano, Dipartimento di Matematica, Via E. Bonardi, 9, 20133 Milano, Italy}
\email{pschlosser@math.tugraz.at}

\author[Daniele C. Struppa]{Daniele C. Struppa}
\address{(DCS) The Donald Bren Presidential Chair in Mathematics, Chapman University, Orange, USA}
\email{struppa@chapman.edu}

\begin{document}

\begin{abstract}
In the last decade there has been a growing interest in superoscillations in various fields of mathematics, physics and engineering. However, while in applications as optics the local oscillatory behaviour is the important property, some convergence to a plane wave is the standard characterizing feature of a superoscillating function in mathematics and quantum mechanics. Also there exists a certain discrepancy between the representation of superoscillations either as generalized Fourier series, as certain integrals or via special functions. The aim of this work is to close these gaps and give a general definition of superoscillations, covering the well-known examples in the existing literature. Superoscillations will be defined as sequences of holomorphic functions, which admit integral representations with respect to complex Borel measures and converge to a plane wave in the space $\mathcal{A}_1(\mathbb{C})$ of exponentially bounded entire functions.
\end{abstract}

\maketitle

\noindent AMS Classification 35A20, 35A08.

\noindent Keywords: Superoscillations, Borel measures, infinite order differential operators.

\section{Introduction}

The theory of superoscillatory functions has its origin in various areas of physics and engineering. For example in quantum mechanics these functions appear in the context of Aharonov's weak values \cite{AAV88}, in antenna theory they were first used in \cite{TDFG} and for optics and other applications we refer to the {\it Roadmap on Superoscillations} by  M.V. Berry et al. \cite{Be19}.

\medskip

The prototypical superoscillatory function that appeared  in quantum mechanics can be found in the still unpublished preprint \cite{APR91} from 1991, where Y. Aharonov and collaborators made a thought experiment considering a box  containing only red light, but emitting a gamma ray. With the notion of weak values the authors also found a way \textit{How the result of a measurement of a component of the spin of a spin-$\frac{1}{2}$ particle can turn out to be 100}, see \cite{AAV88}. The type of functions considered in those papers are of the form
\begin{equation}\label{Eq_Example}
F_n(x)=\sum\limits_{j=0}^nC_j(n)e^{ik_j(n)x},\quad x\in\mathbb{R},\,n\in\mathbb{N},
\end{equation}
where for some fixed $a>1$ the coefficients are chosen as
\begin{equation}\label{Eq_Example_coefficients}
C_j(n)=\Big(\begin{matrix} n \\ j \end{matrix}\Big)\Big(\frac{1+a}{2}\Big)^{n-j}\Big(\frac{1-a}{2}\Big)^j\quad\text{and}\quad k_j(n)=1-\frac{2j}{n}.
\end{equation}
Note that every $F_n$ is a linear combination of plane waves with frequencies $k_j(n)\in[-1,1]$. The superoscillatory behaviour now comes from the fact that
\begin{equation}\label{Eq_Example_convergence}
\lim\limits_{n\rightarrow\infty}F_n(x)=e^{iax},\quad x\in\mathbb{R},
\end{equation}
converges to a plane wave with frequency $a>1$. However, superoscillations of the form \eqref{Eq_Example} do 
not only appear with the frequencies $k_j(n)$ in \eqref{Eq_Example_coefficients}, but it was shown in \cite{ACSSST21} that for any choice of frequencies $k_j(n)\in[-1,1]$ and any $a>1$ one can construct coefficients $C'_j(n)\in\mathbb{C}$, such that
\begin{equation}\label{Eq_Arbitrary_frequencies}
\sum\limits_{j=0}^nC_j'(n)e^{ik_j(n)x}\sim e^{iax},\quad x\in\mathbb{R},
\end{equation}
where $\sim$ is understood in the sense that the Taylor coefficients of both sides coincide up to order $n$.

\medskip

The mathematical theory of superoscillations has attracted a lot of interest in the recent past. A first introduction to the mathematics of superoscillations in one variable and some of its applications to Schr\"odinger evolution of superoscillatory initial data can be found in \cite{ACSST17_1}. Nowadays, the literature on superoscillations is growing quite fast and, without claiming completeness, we mention that some of the most recent results on the time evolution of superoscillations are contained in the papers \cite{ABCS19,ABCS20,Unified,ACSS18,Pozzi,PeterQS} and in the references therein. The case of superoscillating functions in several variables and other interesting mathematical aspects of superoscillations are contained in the papers \cite{ACSST16,tre,ACJSSST22,Talbot,YgerFour}.

\medskip

A certain type of superoscillations, different from the ones in \eqref{Eq_Example}, were introduced by M.V. Berry in his work \textit{Faster than Fourier} \cite{B94}, where families of superoscillating functions are constructed that have the form
\begin{equation}\label{Eq_Berry_superoscillations}
F_\delta(x)=\frac{1}{\delta\sqrt{2\pi}}\int_\mathbb{R}e^{ik(u)x}e^{-\frac{(u-ia)^2}{2\delta^2}}du,\quad x\in\mathbb{R},\,\delta>0,
\end{equation}
with some fixed coefficient $a>0$ and a frequency function $k$ with values $k(u)\in[-1,1]$ for every $u\in\mathbb{R}$. Instead of a convergence result of the form \eqref{Eq_Example_convergence} it is shown in Berry's work that the local wavenumber exceeds the range $[-1,1]$, while the intrinsic frequencies $k$ take only values in $[-1,1]$. 
In the same way also the superoscillatory behaviour of the function
\begin{equation}\label{Eq_Superoscillating_sinc}
F_\delta(x)=\frac{2}{\delta}e^{-\frac{1}{\delta}}\sinc\bigg(\sqrt{x^2-\frac{2iax}{\delta}-\frac{1}{\delta^2}}\bigg)\quad x\in\mathbb{R},\,\delta>0,
\end{equation}
was investigated in \cite{B16}.

\medskip

A method of constructing superoscillating functions was introduced in \cite{ACSST15} and revisited in \cite{ASTY18}. The idea is to consider for some entire function $H(z)=\sum_{n=0}^\infty h_nz^n$ the associated infinite order differential operator $H(-i\frac{\partial}{\partial x})\coloneqq\sum_{n=0}^\infty h_n(-i\frac{\partial}{\partial x})^n$ and the corresponding generalized Schr\"odinger equation
\begin{equation}\label{Eq_Generalized_Schroedinger}
i\frac{\partial}{\partial t}\Psi(t,x)=-H\Big(-i\frac{\partial}{\partial x}\Big)\Psi(t,x),\quad t,x\in\mathbb{R}.
\end{equation}
Using superoscillating initial conditions it turns out that the solution $\Psi(t,\,\cdot\,)$ at any time $t\in\mathbb{R}$ is still superoscillating. Moreover, under some additional assumptions on the range of the function $H$ the solution $\Psi(\,\cdot\,,x)$ is also superoscillating in the time variable for every fixed $x\in\mathbb{R}$.

\medskip

Furthermore, a different perspective on superoscillatory behaviour was taken by P. Ferreira, A. Kempf and D. Lee in a series of papers \cite{FK02,FK06,FKR07,FL14,FL14_1,FL14_2}. There functions of the form
\begin{equation}\label{Eq_Ferreira_Kempf_Lee}
F(x)=\sum\limits_{l=0}^mc_l\sinc(x-x_l),\quad x\in\mathbb{R},
\end{equation}
were considered, where the coefficients $(c_l)_l$ are chosen such that $F(x_l)=a_l$, $l\in\{0,\dots,m\}$, admits the prescribed values $(a_l)_l$ at the prescibed points $(x_l)_l$. Choosing now for example $x_l=\delta l$ for some arbitrary small $\delta>0$ and $a_l=(-1)^l$ with alternating sign, the function $F$ admits an arbitrary large number of oscillations in an arbitrary small interval, while its Fourier transform is always supported in the bounded interval $[-1,1]$.

\medskip

The above examples show that there exist many variants of superoscillations, defined in various ways and having different types of oscillatory behaviour. One of the main purposes of this paper is to propose the general Definition~\ref{defi_Superoscillation} of superoscillations below and to show that all the above examples are contained as special cases in this concept. The crucial points in our analysis are the integral representation  \eqref{Eq_Superoscillating_integral} via complex Borel measures and the space $\mathcal{A}_1(\mathbb{C})$ of exponentially bounded entire functions, where the convergence of superoscillating functions is considered with
respect to the natural topology.

\begin{defi}
The space of entire functions with exponential growth is defined as
\begin{equation}\label{Eq_A1}
\mathcal{A}_1(\mathbb{C})\coloneqq\Set{F:\mathbb{C}\rightarrow\mathbb{C}\text{ entire} | \exists A,B\geq 0\text{ such that }|F(z)|\leq Ae^{B|z|}\text{ for all }z\in\mathbb{C}}.
\end{equation}
For any $F_0,(F_n)_n\in\mathcal{A}_1(\mathbb{C})$ we say that $F_n\rightarrow F_0$ converges in $\mathcal{A}_1(\mathbb{C})$ if and only if there exists some $B\geq 0$, such that
\begin{equation}\label{Eq_A1_convergence}
\lim\limits_{n\rightarrow\infty}\sup\limits_{z\in\mathbb{C}}|F_n(z)-F_0(z)|e^{-B|z|}=0.
\end{equation}
\end{defi}

The idea is now to consider superoscillations as certain superpositions of plane waves with frequencies in a bounded range $[-k_0,k_0]$, but converging to a plane wave with frequency $a\in\mathbb{R}\setminus[-k_0,k_0]$ exceeding this range.

\begin{defi}\label{defi_Superoscillation}
A sequence of functions of the form
\begin{equation}\label{Eq_Superoscillating_integral}
F_n(z)=\int_{-k_0}^{k_0}e^{ikz}d\mu_n(k),\quad z\in\mathbb{C},
\end{equation}
with a common maximal frequency $k_0>0$ and complex Borel measures $\mu_n$ on $[-k_0,k_0]$, is called superoscillating, if there exists some $a\in\mathbb{R}\setminus[-k_0,k_0]$, such that
\begin{equation}\label{Eq_Superoscillating_convergence}
\lim\limits_{n\rightarrow\infty}F_n(z)=e^{iaz}\quad\text{in }\mathcal{A}_1(\mathbb{C}).
\end{equation}
\end{defi}

Note, that any function of the form \eqref{Eq_Superoscillating_integral} is an element in $\mathcal{A}_1(\mathbb{C})$. In fact, since a complex measure has finite total variation, the exponential boundedness follows immediately from the estimate
\begin{equation*}
|F_n(z)|\leq\int_{-k_0}^{k_0}e^{|k||z|}d|\mu_n|(z)\leq|\mu_n|([-k_0,k_0])e^{k_0|z|},\quad z\in\mathbb{C}.
\end{equation*}
Moreover, the holomorphicity follows from a version of the dominated convergence theorem which allows to carry the derivatives inside the integral.

\medskip

Also note that the functions \eqref{Eq_Example} are superoscillation in the sense of Definition~\ref{defi_Superoscillation} since the representation \eqref{Eq_Superoscillating_integral} holds 
with the complex Borel measures
\begin{equation}\label{Eq_Point_measure}
\mu_n(B)\coloneqq\sum\limits_{j=0,k_j(n)\in B}^nC_j(n),\quad\text{for every Borel set }B\subseteq[-1,1],
\end{equation}
and the convergence \eqref{Eq_Superoscillating_convergence} was  shown in \cite[Theorem 2.1]{CSSY22}.

\medskip

In the main part of this paper we collect and generalize the different superoscillating functions that appear in the mathematical and physical literature, and prove that they fit into the context of Definition~\ref{defi_Superoscillation}. In particular, we verify the convergence in the space $\mathcal{A}_1(\mathbb{C})$ in all cases, which is important from an analytic point of view and puts the phenomenon of superoscillations in an appropriate mathematical perspective. In Section~\ref{sec_Berry_superoscillations} functions of the form \eqref{Eq_Berry_superoscillations} and certain extensions of these classes are discussed. The superoscillating $\sinc$-function \eqref{Eq_Superoscillating_sinc} is treated in Section~\ref{sec_Superoscillating_sinc_function}. Section~\ref{sec_Superoscillations_with_prescribed_frequencies} shows that the approximation \eqref{Eq_Arbitrary_frequencies} really leads to superoscillations in the sense of Definition~\ref{defi_Superoscillation} for a large class of given frequencies $k_j(n)$, see Corollary~\ref{cor_Fn_exponentials}. Moreover, in Theorem~\ref{satz_Fn_convergence} this concept is generalized to any low frequency function of the form \eqref{Eq_fjn2} instead of the exponentials $e^{ik_j(n)z}$. It then turns out that a version of the method \eqref{Eq_Ferreira_Kempf_Lee} can be considered as a special case of Corollary \ref{cor_Fn_derivative_convergence}, treated in 
Example~\ref{bsp_Ferreira_Kempf_Lee}. The final Section~\ref{sec_Superoscillating_solutions_of_the_generalized_Schroedinger_equation} then considers the method of generating superoscillating functions described in \eqref{Eq_Generalized_Schroedinger}.

\medskip

\noindent\textbf{Acknowledgements}.
Jussi Behrndt gratefully acknowledges financial support by the Austrian Science Fund (FWF): P 33568-N. Peter Schlosser's research was funded by the Austrian Science Fund (FWF) under Grant No. J 4685-N and by the European Union -- NextGenerationEU. This publication is also based upon work from COST Action CA 18232 MAT-DYN-NET, supported by COST (European Cooperation in Science and Technology), www.cost.eu.

\section{A construction of superoscillations due to M.V. Berry}\label{sec_Berry_superoscillations}

The considerations in this section are inspired by M.V. Berry's paper \cite{B94}, where superoscillatory functions of the form
\begin{equation}\label{abc}
F_\delta(x)=\frac{1}{\delta\sqrt{2\pi}}\int_\mathbb{R}e^{ik(u)x}e^{-\frac{(u-ia)^2}{2\delta^2}}du,\quad x\in\mathbb{R},
\end{equation}
were introduced. The idea behind this construction is that the complex Gaussian
\begin{equation*}
\frac{1}{\delta\sqrt{2\pi}}e^{-\frac{(u-ia)^2}{2\delta^2}}\rightarrow\delta(u-ia),\quad\text{as }\delta\rightarrow 0^+,
\end{equation*}
approximates the complex delta function and consequently
\begin{equation*}
F_\delta(x)\rightarrow e^{ik(ia)x},\quad\text{as }\delta\rightarrow 0^+,
\end{equation*}
converges to a plane wave. If one chooses a frequency function with values $k(u)\in[-1,1]$ for every $u\in\mathbb{R}$ and $k(ia)\in\mathbb{R}\setminus[-1,1]$ for some $a>0$, this indicates a superoscillatory behaviour of the functions $F_\delta$.

\medskip

We will now revisit this idea and improve it in three ways. Firstly, an additional function $g$ is included in the integral \eqref{abc}. This function does not affect the superoscillatory property of the $F_\delta$'s, but allows to modify their shape. Secondly, precise assumptions on the involved functions $g$ and $k$ are given. Thirdly and most importantly, while in \cite{B94} mainly the complex saddle point approximation is used to derive quantities like the local wave number, the convergence in the space $\mathcal{A}_1(\mathbb{C})$ is proven here. In particular, Corollary~\ref{cor_Fdelta} shows that the functions $F_\delta$ are indeed superoscillating in the sense of Definition~\ref{defi_Superoscillation}.

\begin{satz}\label{satz_Fdelta}
For $b\geq a>0$ let $\Delta\subseteq\mathbb{C}$ be the open triangle with corners $ia$ and $\pm b$. Consider measurable functions $k,g:\mathbb{R}\cup\overline{\Delta}\rightarrow\mathbb{C}$ which for some $B\geq 0$ satisfy the bounds
\begin{equation}\label{Eq_Fdelta_assumptions}
\sup\limits_{u\in\mathbb{R}}|k(u)|<\infty\quad\text{and}\quad\sup\limits_{u\in\mathbb{R}}|g(u)|e^{-B|u|}<\infty.
\end{equation}
Moreover, assume that $k$ and $g$ are continuous on $\overline{\Delta}$, holomorphic on $\Delta$, and in the case $a=b$ their derivatives are bounded on $\Delta$. Then, for every $\delta>0$, the functions
\begin{equation}\label{Eq_Fdelta}
F_\delta(z)\coloneqq\frac{1}{\delta\sqrt{2\pi}}\int_\mathbb{R}g(u)e^{ik(u)z}e^{-\frac{(u-ia)^2}{2\delta^2}}du,\quad z\in\mathbb{C},
\end{equation}
belong to the space $\mathcal{A}_1(\mathbb{C})$ and converge as
\begin{equation}\label{Eq_Fdelta_convergence}
\lim\limits_{\delta\rightarrow 0^+}F_\delta(z)=g(ia)e^{ik(ia)z}\quad\text{in }\mathcal{A}_1(\mathbb{C}).
\end{equation}
\end{satz}

\begin{proof}
We shall make use of the constants
\begin{equation*}
C\coloneqq\sup\limits_{u\in\mathbb{R}}|k(u)|\quad\text{and}\quad A\coloneqq\sup\limits_{u\in\mathbb{R}}|g(u)|e^{-B|u|},
\end{equation*}
which are finite due to assumption \eqref{Eq_Fdelta_assumptions}. A straightforward estimate shows
\begin{equation}\label{esti1}
|F_\delta(z)|\leq\frac{1}{\delta\sqrt{2\pi}}\int_\mathbb{R}|g(u)|e^{|k(u)||z|}e^{\frac{-u^2+a^2}{2\delta^2}}du\leq\frac{A}{\delta\sqrt{2\pi}}e^{\frac{a^2}{2\delta^2}} \bigg(\int_\mathbb{R}e^{B|u|}e^{-\frac{u^2}{2\delta^2}}du \bigg)e^{C|z|}, 
\end{equation}
and hence $F_\delta$ is well defined for every $\delta>0$. Furthermore, a version of the dominated convergence 
theorem shows that $F_\delta$ is holomorphic and together with the above estimate it follows that $F_\delta$, 
$\delta>0$,
belongs to the space $\mathcal{A}_1(\mathbb{C})$.

\medskip

In the following we will verify that $F_\delta$ converges as in \eqref{Eq_Fdelta_convergence}. First, we change the part $-b\rightarrow b$ of the integration path in \eqref{Eq_Fdelta} to $-b\rightarrow ia\rightarrow b$, that is, we use
\begin{equation*}
\int_{-b}^bg(u)e^{ik(u)z}e^{-\frac{(u-ia)^2}{2\delta^2}}du=\int_{-b}^{ia}g(u)e^{ik(u)z}e^{-\frac{(u-ia)^2}{2\delta^2}}du+\int_{ia}^bg(u)e^{ik(u)z}e^{-\frac{(u-ia)^2}{2\delta^2}}du.
\end{equation*}
In fact,
since $g$ and $k$ are holomorphic on $\Delta$ and continuous on $\overline{\Delta}$ the above equality follows 
from the Cauchy theorem applied to the boundary of the scaled triangle $\Delta_\varepsilon$ and then taking the limit $\varepsilon\rightarrow 0^+$.
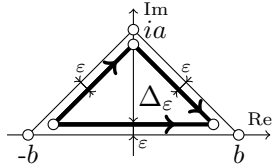
\begin{center}
\begin{tikzpicture}[scale=1.4]
\draw[->] (-1.2,0)--(1.3,0);
\draw (1.2,0) node[anchor=south] {\tiny{$\Re$}};
\draw[->] (0,-0.2)--(0,1.2) node[anchor=west] {\tiny{$\Im$}};
\draw (-1,0)--(0,1)--(1,0);
\draw[ultra thick] (-0.76,0.1)--(0,0.86)--(0.76,0.1)--(-0.76,0.1);
\draw[ultra thick,->] (-0.76,0.1)--(-0.16,0.7);
\draw[ultra thick,->] (0,0.86)--(0.65,0.21);
\draw[ultra thick,->] (-0.76,0.1)--(0.4,0.1);
\draw[fill=white] (-1,0) circle (0.05) node[anchor=north] {\small{-$b$}};
\draw[fill=white] (1,0) circle (0.05) node[anchor=north] {\small{$b$}};
\draw[fill=white] (0,1) circle (0.05) node[anchor=west] {\small{$ia$}};
\draw[>-<] (0,-0.05)--(0,0.165);
\draw (-0.05,0.05) node[anchor=north west] {\tiny{$\varepsilon$}};
\draw[>-<] (-0.535,0.535)--(-0.384,0.384);
\draw (-0.55,0.55)--(-0.38,0.38);
\draw (-0.5,0.5) node[anchor=south] {\tiny{$\varepsilon$}};
\draw[>-<] (0.535,0.535)--(0.384,0.384);
\draw (0.55,0.55)--(0.38,0.38);
\draw (0.5,0.5) node[anchor=south] {\tiny{$\varepsilon$}};
\draw[fill=white] (-0.76,0.1) circle (0.05);
\draw[fill=white] (0.76,0.1) circle (0.05);
\draw[fill=white] (0,0.86) circle (0.05);
\draw (-0.05,0.35) node[anchor=west] {$\Delta_\varepsilon$};
\end{tikzpicture}
\end{center}
Hence we can split up the function $F_\delta$ into the four parts
\begin{equation}\label{Eq_Fdelta_12}
 F_\delta(z)=F_\delta^{(1)}(z)+F_\delta^{(2)}(z)+F_\delta^{(3)}(z)+F_\delta^{(4)}(z),\quad z\in\mathbb{C},
\end{equation}
where 
\begin{equation*}
 \begin{split}
F_\delta^{(1)}(z)&=\frac{1}{\delta\sqrt{2\pi}}\int_{-\infty}^{-b}g(u)e^{ik(u)z}e^{-\frac{(u-ia)^2}{2\delta^2}}du,\\
F_\delta^{(2)}(z)&=\frac{1}{\delta\sqrt{2\pi}}\int_{-b}^{ia}g(u)e^{ik(u)z}e^{-\frac{(u-ia)^2}{2\delta^2}}du,\\
F_\delta^{(3)}(z)&=\frac{1}{\delta\sqrt{2\pi}}\int_{ia}^{b}g(u)e^{ik(u)z}e^{-\frac{(u-ia)^2}{2\delta^2}}du,\\
F_\delta^{(4)}(z)&=\frac{1}{\delta\sqrt{2\pi}}\int_{b}^{\infty}g(u)e^{ik(u)z}e^{-\frac{(u-ia)^2}{2\delta^2}}du.
 \end{split}
\end{equation*}
For these four functions we will now investigate the limit $\delta\rightarrow 0^+$ separately. Starting with $F_\delta^{(4)}$ we estimate in the same way as in \eqref{esti1}
\begin{equation*}
|F_\delta^{(4)}(z)|\leq\frac{A}{\delta\sqrt{2\pi}}e^{C|z|}e^{\frac{a^2}{2\delta^2}}\int_b^\infty e^{Bu}e^{-\frac{u^2}{2\delta^2}}du\leq\frac{A}{\sqrt{2\pi}}e^{C|z|}e^{bB}\int_0^\infty e^{-\frac{v^2}{2}}e^{v(B\delta-\frac{b}{\delta})}dv,
\end{equation*}
where in the second inequality we used $a\leq b$ and substituted $v=\frac{u-b}{\delta}$. Shifting the exponential $e^{C|z|}$ to the left side of this inequality leads to a $z$-independent right hand side, which 
converges as
\begin{equation*}
\sup\limits_{z\in\mathbb{C}}|F_\delta^{(4)}(z)|e^{-C|z|}\leq\frac{A}{\sqrt{2\pi}}e^{bB}\int_0^\infty e^{-\frac{v^2}{2}}e^{v(B\delta-\frac{b}{\delta})}dv\rightarrow 0,\quad\text{as }\delta\rightarrow 0^+.
\end{equation*}
According to \eqref{Eq_A1_convergence} this is exactly the convergence
\begin{equation*}
\lim\limits_{\delta\rightarrow 0^+}F_\delta^{(4)}(z)=0\quad\text{in }\mathcal{A}_1(\mathbb{C}).
\end{equation*}
In the same way we also prove
\begin{equation*}
\lim\limits_{\delta\rightarrow 0^+}F_\delta^{(1)}(z)=0\quad\text{in }\mathcal{A}_1(\mathbb{C}).
\end{equation*}
For the function $F_\delta^{(3)}(z)$ we first use the identity $$
\int_{ia}^be^{-\frac{(u-ia)^2}{2\delta^2}}du=\frac{\delta\sqrt{\pi}}{\sqrt{2}}\erf\Big(\frac{b-ia}{\delta\sqrt{2}}\Big),
$$ 
which is an immediate consequence of the definition of the error function, to rewrite the difference
\begin{equation}\label{Eq_Fdelta3_difference}
F_\delta^{(3)}(z)-\frac{g(ia)}{2}e^{ik(ia)z}=\frac{1}{\delta\sqrt{2\pi}}\int_{ia}^b\Big(g(u)e^{ik(u)z}-\frac{g(ia)}{\erf\big(\frac{b-ia}{\delta\sqrt{2}}\big)}e^{ik(ia)z}\Big)e^{-\frac{(u-ia)^2}{2\delta^2}}du.
\end{equation}
In the case $b>a$ we use $u_\delta(t)\coloneqq ia+\delta t(b-ia)$, $t\in[0,\frac{1}{\delta}]$, to parametrize this integral as
\begin{equation*}
F_\delta^{(3)}(z)-\frac{g(ia)}{2}e^{ik(ia)z}=\frac{b-ia}{\sqrt{2\pi}}\int_0^{\frac{1}{\delta}}\Big(g(u_\delta(t))e^{ik(u_\delta(t))z}-\frac{g(ia)}{\erf\big(\frac{b-ia}{\delta\sqrt{2}}\big)}e^{ik(ia)z}\Big)e^{-\frac{t^2(b-ia)^2}{2}}dt,
\end{equation*}
and estimate the difference by
\begin{align*}
\Big|F_\delta^{(3)}(z)-\frac{g(ia)}{2}e^{ik(ia)z}\Big|\leq\frac{\sqrt{b^2+a^2}}{\sqrt{2\pi}}\int_0^{\frac{1}{\delta}}\bigg(&|g(u_\delta(t))|\,\big|e^{ik(u_\delta(t))z}-e^{ik(ia)z}\big| \\
&+\bigg|g(u_\delta(t))-\frac{g(ia)}{\erf(\frac{b-ia}{\delta\sqrt{2}})}\bigg|e^{|k(ia)|\,|z|}\bigg)e^{-\frac{t^2(b^2-a^2)}{2}}dt.
\end{align*}
Since $g,k$ are continuous on $\overline{\Delta}$, we denote their respective suprema by $\Vert\cdot\Vert_\mathsmaller{\Delta}$. Using the estimate \eqref{Eq_Exponential_estimate} we obtain
\begin{align*}
\Big|F_\delta^{(3)}(z)-\frac{g(ia)}{2}e^{ik(ia)z}\Big|\leq\frac{\sqrt{b^2+a^2}}{\sqrt{2\pi}}e^{\Vert k\Vert_\mathsmaller{\Delta}|z|}\int_0^{\frac{1}{\delta}}\bigg(&|k(u_\delta(t))-k(ia)|\,\Vert g\Vert_\mathsmaller{\Delta}|z| \\
&+\bigg|g(u_\delta(t))-\frac{g(ia)}{\erf(\frac{b-ia}{\delta\sqrt{2}})}\bigg|\bigg)e^{-\frac{t^2(b^2-a^2)}{2}}dt.
\end{align*}
Using also the inequality  $|z|\leq e^{|z|}$, which is an immediate consequence of the power series representation of the exponential, we find the estimate
\begin{align*}
\Big|F_\delta^{(3)}(z)-\frac{g(ia)}{2}e^{ik(ia)z}\Big|e^{-(\Vert k\Vert_\mathsmaller{\Delta}+1)|z|}\leq\frac{\sqrt{b^2+a^2}}{\sqrt{2\pi}}\int_0^{\frac{1}{\delta}}&\bigg(|k(u_\delta(t))-k(ia)|\,\Vert g\Vert_\mathsmaller{\Delta} \\
&+\bigg|g(u_\delta(t))-\frac{g(ia)}{\erf(\frac{b-ia}{\delta\sqrt{2}})}\bigg|\bigg)e^{-\frac{t^2(b^2-a^2)}{2}}dt.
\end{align*}
In this form we note that the right hand side is independent of $z$ and vanishes as $\delta\rightarrow 0^+$ due to the dominated convergence theorem, the continuity of $g$ and $k$ and the fact that $\lim_{\delta\rightarrow 0^+}\erf(\frac{b-ia}{\delta\sqrt{2}})=1$ by the choice $b>a$. This shows the convergence
\begin{equation*}
\lim\limits_{\delta\rightarrow 0^+}F_\delta^{(3)}(z)=\frac{g(ia)}{2}e^{ik(ia)z}\quad\text{in }\mathcal{A}_1(\mathbb{C}).
\end{equation*}
In the case $a=b$ we go back to formula \eqref{Eq_Fdelta3_difference} and note that all functions in the integrand are continuous. Hence we can shift the integration path into the open triangle $\Delta$ where $g$ and $k$ are holomorphic, i.e. we write the integral as the limit
\begin{equation*}
F_\delta^{(3)}(z)-\frac{g(ia)}{2}e^{ik(ia)z}=\frac{1}{\delta\sqrt{2\pi}}\lim\limits_{\varepsilon\rightarrow 0^+}\int_{i(a-\varepsilon)}^{a-\varepsilon}\Big(g(u)e^{ik(u)z}-\frac{g(ia)}{\erf\big(\frac{a(1-i)}{\delta\sqrt{2}}\big)}e^{ik(ia)z}\Big)e^{-\frac{(u-ia)^2}{2\delta^2}}du.
\end{equation*}
Using the complementary error function $\erfc(z)\coloneqq 1-\erf(z)$ and its derivative
\begin{equation*}
\frac{d}{du}\erfc\Big(\frac{u-ia}{\delta\sqrt{2}}\Big)=-\frac{\sqrt{2}}{\delta\sqrt{\pi}}e^{-\frac{(u-ia)^2}{2\delta^2}},
\end{equation*}
integration by parts leads to
\begin{align*}
&F_\delta^{(3)}(z)-\frac{g(ia)}{2}e^{ik(ia)z} \\
&\quad=-\frac{1}{2}\lim\limits_{\varepsilon\rightarrow 0^+}\int_{i(a-\varepsilon)}^{a-\varepsilon}\Big(g(u)e^{ik(u)z}-\frac{g(ia)}{\erf\big(\frac{a(1-i)}{\delta\sqrt{2}}\big)}e^{ik(ia)z}\Big)\frac{d}{du}\erfc\Big(\frac{u-ia}{\delta\sqrt{2}}\Big)du \\
&\quad=-\frac{g(a)}{2}e^{ik(a)z}\erfc\Big(\frac{a(1-i)}{\delta\sqrt{2}}\Big)+\frac{1}{2}\lim\limits_{\varepsilon\rightarrow 0^+}\int_{i(a-\varepsilon)}^{a-\varepsilon}\frac{d}{du}\Big(g(u)e^{ik(u)z}\Big)\erfc\Big(\frac{u-ia}{\delta\sqrt{2}}\Big)du \\
&\quad=-\frac{g(a)}{2}e^{ik(a)z}\erfc\Big(\frac{a(1-i)}{\delta\sqrt{2}}\Big)+\frac{1}{2}\lim\limits_{\varepsilon\rightarrow 0^+}\int_0^{a-\varepsilon}\frac{d}{ds}\Big(g(u_\varepsilon(s))e^{ik(u_\varepsilon(s))z}\Big)\erfc\Big(\frac{u_\varepsilon(s)-ia}{\delta\sqrt{2}}\Big)ds,
\end{align*}
where in the last line we parametrized the complex path integral by 
$u_\varepsilon(s)\coloneqq i(a-\varepsilon)+s(1-i)$, $s\in [0,a-\varepsilon]$.  Using this, we can now estimate the difference by
\begin{align}
\Big|F_\delta^{(3)}(z)-\frac{g(ia)}{2}e^{ik(ia)z}\Big|&\leq e^{\Vert k\Vert_\mathsmaller{\Delta}|z|}\bigg(\frac{\Vert g\Vert_\mathsmaller{\Delta}}{2}\Big|\erfc\Big(\frac{a(1-i)}{\delta\sqrt{2}}\Big)\Big| \notag\\
&\hspace{2cm}+\frac{\Vert g'\Vert_\mathsmaller{\Delta}+\Vert gk'\Vert_\mathsmaller{\Delta}|z|}{\sqrt{2}}\lim\limits_{\varepsilon\rightarrow 0^+}\int_0^{a-\varepsilon}\Big|\erfc\Big(\frac{u_\varepsilon(s)-ia}{\delta\sqrt{2}}\Big)\Big|ds\bigg) \notag\\
&\leq e^{(\Vert k\Vert_\mathsmaller{\Delta}+1)|z|}\bigg(\frac{\Vert g\Vert_\mathsmaller{\Delta}}{2}\Big|\erfc\Big(\frac{a(1-i)}{\delta\sqrt{2}}\Big)\Big| \notag\\
&\hspace{2.5cm}+\frac{\Vert g'\Vert_\mathsmaller{\Delta}+\Vert gk'\Vert_\mathsmaller{\Delta}}{\sqrt{2}}\int_0^a\Big|\erfc\Big(\frac{s(1-i)}{\delta\sqrt{2}}\Big)\Big|ds\bigg), \label{Eq_Fdelta3}
\end{align}
where in the second inequality we performed the limit $\varepsilon\rightarrow 0^+$ and used the estimate 
$|z|\leq e^{|z|}$ to get rid of the $z$-dependency inside the brackets. Using the representation $\erfc(z)=\frac{2}{\sqrt{\pi}}\int_0^\infty e^{-(t+z)^2}dt$ of the complementary error function, see \cite[Eq. (7.1.2)]{AS72}, it can be estimated by
\begin{equation*}
\Big|\erfc\Big(\frac{s(1-i)}{\delta\sqrt{2}}\Big)\Big|=\frac{2}{\sqrt{\pi}}\Big|\int_0^\infty e^{-(t+\frac{s(1-i)}{\delta\sqrt{2}})^2}dt\Big|\leq\frac{2}{\sqrt{\pi}}\int_0^\infty e^{-t^2-\frac{\sqrt{2}\,ts}{\delta}}dt\leq\min\Big\{\frac{\delta\sqrt{2}}{s\sqrt{\pi}},1\Big\},
\end{equation*}
where in the last inequality we estimated either $e^{-t^2}\leq 1$ or $e^{-\frac{\sqrt{2}\,ts}{\delta}}\leq 1$ to get the minimum as an upper bound. This inequality on the one hand gives an $\delta$-independent bound of the last integral in \eqref{Eq_Fdelta3} as well as for every $s\in(0,a]$ the pointwise convergence as $\delta\rightarrow 0^+$. With the dominated convergence theorem this then leads to
\begin{align*}
\Big|F_\delta^{(3)}(z)-\frac{g(ia)}{2}e^{ik(ia)z}\Big|e^{-(\Vert k\Vert_\mathsmaller{\Delta}+1)|z|}\leq\,&\frac{\Vert g\Vert_\mathsmaller{\Delta}}{2}\Big|\erfc\Big(\frac{a(1-i)}{\delta\sqrt{2}}\Big)\Big| \\
&+\frac{\Vert g'\Vert_\mathsmaller{\Delta}+\Vert gk'\Vert_\mathsmaller{\Delta}}{\sqrt{2}}\int_0^a\Big|\erfc\Big(\frac{s(1-i)}{\delta\sqrt{2}}\Big)\Big|ds\overset{\delta\rightarrow 0^+}{\longrightarrow}0.
\end{align*}
Since the right hand side is independent of $|z|$ we conclude the convergence
\begin{equation*}
\lim\limits_{\delta\rightarrow 0^+}F_\delta^{(3)}(z)=\frac{g(ia)}{2}e^{ik(ia)z}\quad\text{in }\mathcal{A}_1(\mathbb{C}).
\end{equation*}
For the same reason also
\begin{equation*}
\lim\limits_{\delta\rightarrow 0^+}F_\delta^{(2)}(z)=\frac{g(ia)}{2}e^{ik(ia)z}\quad\text{in }\mathcal{A}_1(\mathbb{C}).
\end{equation*}
Summing up, we have proved that all the terms in \eqref{Eq_Fdelta_12} converge in $\mathcal{A}_1(\mathbb{C})$
and hence \eqref{Eq_Fdelta_convergence} follows. 
\end{proof}

The next corollary puts the result of Theorem~\ref{satz_Fdelta} into the perspective of superoscillations and proves that under certain assumptions on the frequency function $k$, the resulting functions $F_\delta$ satisfy Definition~\ref{defi_Superoscillation}.

\begin{cor}\label{cor_Fdelta}
Let $b\geq a>0$ and $k,g$ be as in Theorem~\ref{satz_Fdelta}. Assume, in addition, that there exists some $k_0>0$ 
such that
\begin{equation}\label{Eq_Fdelta_superoscillating_assumptions}
g(ia)=1,\quad k(ia)\in\mathbb{R}\setminus[-k_0,k_0]\quad\text{and}\quad k(u)\in[-k_0,k_0]\text{ for every }u\in\mathbb{R}.
\end{equation}
Then the functions $F_\delta$ in \eqref{Eq_Fdelta} are superoscillating with limit
\begin{equation}\label{Eq_Fdelta_2}
\lim\limits_{\delta\rightarrow 0^+}F_\delta(z)=e^{ik(ia)z}\quad\text{in }\mathcal{A}_1(\mathbb{C}).
\end{equation}
\end{cor}

\begin{proof}
Since $g(ia)=1$ by assumption, the convergence \eqref{Eq_Fdelta_2} follows from Theorem~\ref{satz_Fdelta}. Hence it remains to verify the integral representation \eqref{Eq_Superoscillating_integral}. For this we first define the complex Borel measure
\begin{equation*}
\sigma_\delta(A)\coloneqq\frac{1}{\delta\sqrt{2\pi}}\int_Ag(u)e^{-\frac{(u-ia)^2}{2\delta^2}}du
\end{equation*}
for Borel sets $A\subseteq\mathbb{R}$.
With this measure the function $F_\delta$ admits the representation
\begin{equation*}
F_\delta(z)=\int_\mathbb{R}e^{ik(u)z}d\sigma_\delta(u),\quad z\in\mathbb{C}.
\end{equation*}
In a second step we consider the Borel measure
\begin{equation*}
\mu_\delta(B)\coloneqq\sigma_\delta\big(\Set{u\in\mathbb{R} | k(u)\in B}\big)
\end{equation*}
for Borel sets $B\subseteq[-k_0,k_0]$ and rewrite
$F_\delta$ as
\begin{equation*}
F_\delta(z)=\int_{-k_0}^{k_0}e^{ikz}d\mu_\delta(k),\quad z\in\mathbb{C},
\end{equation*}
which is exactly the form \eqref{Eq_Superoscillating_integral}. Since by assumption $k(u)\in[-k_0,k_0]$ for every $u\in\mathbb{R}$ and since the frequency $k(ia)$ of the limit function in \eqref{Eq_Fdelta_2} lies in $\mathbb{R}\setminus[-k_0,k_0]$, the functions $F_\delta$ are indeed superoscillating in the sense of Definition~\ref{defi_Superoscillation}.
\end{proof}

\begin{bem}
Note that in Theorem~\ref{satz_Fdelta} and Corollary~\ref{cor_Fdelta} we may modify the functions $k,g$ on $\mathbb{R}\setminus[-b,b]$ as long as \eqref{Eq_Fdelta_assumptions} and \eqref{Eq_Fdelta_superoscillating_assumptions} remain valid. Since this only changes the functions $F_\delta$ but not the limit $g(ia)e^{ik(ia)z}$ in \eqref{Eq_Fdelta_convergence}, any such modification leads to a new family of superoscillating functions.
\end{bem}

To illustrate our result we provide some possible choices of functions $k$ and $g$ which lead to superoscillations in Corollary~\ref{cor_Fdelta}. In particular, we will use the functions $k$ from the original paper \cite{B94}.

\begin{bsp}
Possible choices of the frequency function are
\begin{equation*}
k_1(u)=\frac{1}{1+\frac{u^2}{2}},\quad k_2(u)=\frac{1}{\cosh(u)},\quad k_3(u)=e^{-\frac{u^2}{2}},\quad k_4(u)=\cos(u).
\end{equation*}
All these functions satisfy $k_{j}(u)\in[-1,1]$ for $j=1,2,3,4$ and for every $u\in\mathbb{R}$. Evaluated for complex arguments these functions admit the values
\begin{equation*}
k_1(ia)=\frac{1}{1-\frac{a^2}{2}},\quad k_2(ia)=\frac{1}{\cos(a)},\quad k_3(ia)=e^{\frac{a^2}{2}},\quad k_4(ia)=\cosh(a),
\end{equation*}
and hence $k_1(ia)>1$ for every $a\in(0,\sqrt{2})$, $k_2(ia)>1$ for every $a\in(0,\frac{\pi}{2})$, and $k_{3,4}(ia)>1$ for every $a>0$. With the constant $g_{j}\equiv 1$ for $j=1,2,3,4$, the assumptions of Corollary~\ref{cor_Fdelta} are satisfied and we end up with superoscillating functions $F_\delta$. Another possible choice considered in \cite{B94} is for any $0<a\leq 2$ the functions
\begin{equation*}
g_5(u)=\begin{cases} 1, & \text{if }u\in\overline{\Delta}, \\ 0, & \text{if }u\in\mathbb{R}\setminus[-a,a], \end{cases}\quad\text{and}\quad k_5(u)=\begin{cases} 1-\frac{u^2}{2}, & \text{if }u\in\overline{\Delta} \\ 0, & \text{if }u\in\mathbb{R}\setminus[-a,a]. \end{cases}
\end{equation*}
Here, $\Delta$ is the triangle from Theorem~\ref{satz_Fdelta} with $b=a$. Also these functions satisfy the assumtions of Corollary~\ref{cor_Fdelta} and lead to superoscillating functions $F_\delta$.
\end{bsp}

\section{Superoscillating sinc-functions}\label{sec_Superoscillating_sinc_function}

In \cite{B16} M.V. Berry considered another type of superoscillating functions 
\begin{equation}\label{Eq_Fdelta_sinc}
F_\delta(z)=\frac{2}{\delta}e^{-\frac{1}{\delta}}\sinc\bigg(\sqrt{z^2-\frac{2iaz}{\delta}-\frac{1}{\delta^2}}\bigg),\quad z\in\mathbb{C},
\end{equation}
for some fixed $a>1$ and every $\delta>0$; here $\sinc(z)\coloneqq\frac{\sin(z)}{z}$ denotes 
the Sinus cardinalis and the complex square root is fixed by
\begin{equation}\label{Eq_Root}
0\leq\Arg(\sqrt{\,\cdot\,}\,)<\pi.
\end{equation}
The aim of this section is to verify that \eqref{Eq_Fdelta_sinc} is indeed superoscillatory in the sense of Definition~\ref{defi_Superoscillation}. This will be done in two steps. In Theorem~\ref{satz_Fdelta_sinc} we prove the convergence \eqref{Eq_Superoscillating_convergence} and with the help of Lemma~\ref{lem_Integral_representation_of_sinc} we conclude the integral representation \eqref{Eq_Superoscillating_integral}.

\medskip

We start with a technical estimate of the complex square root, which will be important for the $\mathcal{A}_1$-convergence of the functions \eqref{Eq_Fdelta_sinc} in Theorem~\ref{satz_Fdelta_sinc}.

\begin{lem}\label{lem_Root_estimate}
For any $a>1$ there exists some $C\geq 0$, such that
\begin{equation}\label{Eq_Root_estimate}
\big|\sqrt{z^2-2iaz-1}+az-i\big|\leq C\min\{|z|,|z|^2\},\quad z\in\mathbb{C}.
\end{equation}
\end{lem}

\begin{proof}
In the first step we consider the function $f(z)\coloneqq\sqrt{z^2-2iaz-1}-az+i$. Assume that $f(z_0)=0$ for some $z_0\in\mathbb{C}$, that is,
\begin{equation*}
\sqrt{z_0^2-2iaz_0-1}=az_0-i.
\end{equation*}
Squaring both sides leads to $z_0^2=a^2z_0^2$ and since obviously $z_0\neq 0$ by the choice of the square root in \eqref{Eq_Root}, this is a contradiction to $a>1$. Hence $f$ does not have any zeros. Since $a>1$ by assumption 
we also conclude
\begin{equation*}
\lim\limits_{|z|\rightarrow\infty}|f(z)|=\infty,
\end{equation*}
and consequently there exists $c>0$ such that
\begin{equation}\label{Eq_Root_estimate_3}
|f(z)|\geq c,\quad z\in\mathbb{C}.
\end{equation}
In the second step we consider $g(z)\coloneqq\sqrt{z^2-2iaz-1}+az-i$. Using \eqref{Eq_Root_estimate_3}, this function can be estimated by
\begin{equation}\label{Eq_Root_estimate_1}
|g(z)|=\Big|\frac{g(z)f(z)}{f(z)}\Big|=\Big|\frac{(1-a^2)z^2}{f(z)}\Big|\leq\frac{(a^2-1)|z|^2}{c},\quad z\in\mathbb{C}.
\end{equation}
Moreover, $g$ can also be estimated by
\begin{equation}\label{Eq_Root_estimate_2}
|g(z)|\leq\sqrt{(|z|+a)^2}+a|z|+1=(1+a)(1+|z|),\quad z\in\mathbb{C}.
\end{equation}
Using now \eqref{Eq_Root_estimate_1} for $|z|\leq 1$ and \eqref{Eq_Root_estimate_2} for $|z|\geq 1$ this implies the estimate \eqref{Eq_Root_estimate}.
\end{proof}

\begin{satz}\label{satz_Fdelta_sinc}
For every $a>1$ and $\delta>0$ the functions $F_\delta$ in \eqref{Eq_Fdelta_sinc} belong to the space 
$\mathcal{A}_1(\mathbb{C})$ and converge as
\begin{equation}\label{Eq_Fdelta_sinc_convergence}
\lim\limits_{\delta\rightarrow 0^+}F_\delta(z)=e^{iaz}\quad\text{in }\mathcal{A}_1(\mathbb{C}).
\end{equation}
\end{satz}

\begin{proof}
For this proof it will be convenient to use the notation
\begin{equation}\label{Eq_Rdelta}
R_\delta(z)\coloneqq\sqrt{z^2-\frac{2iaz}{\delta}-\frac{1}{\delta^2}},\quad z\in\mathbb{C},
\end{equation}
so that we can write $F_\delta(z)=\frac{2}{\delta}e^{-\frac{1}{\delta}}\sinc(R_\delta(z))$. Although the complex square root in \eqref{Eq_Rdelta} is not entire, in the series expansion of the $\sinc$-function only even powers appear and we still end up with the entire function
\begin{equation*}
F_\delta(z)=\frac{2}{\delta}e^{-\frac{1}{\delta}}\sum\limits_{n=0}^\infty\frac{(-1)^n}{(2n+1)!}R_\delta(z)^{2n}=\frac{2}{\delta}e^{-\frac{1}{\delta}}\sum\limits_{n=0}^\infty\frac{(-1)^n}{(2n+1)!}\Big(z^2-\frac{2iaz}{\delta}-\frac{1}{\delta^2}\Big)^n.
\end{equation*}
For $\delta>0$ fixed and $z\in\mathbb{C}$ sufficiently large we have $1\leq|R_\delta(z)|\leq 2|z|$ and hence the exponential bound
\begin{equation*}
|F_\delta(z)|=\frac{2}{\delta}e^{-\frac{1}{\delta}}\frac{|\sin(R_\delta(z))|}{|R_\delta(z)|}\leq\frac{1}{\delta}e^{-\frac{1}{\delta}}\big|e^{iR_\delta(z)}-e^{-iR_\delta(z)}\big|\leq\frac{2}{\delta}e^{-\frac{1}{\delta}}e^{2|z|}
\end{equation*}
holds for all $z\in\mathbb{C}$ sufficiently large. It follows that $F_\delta\in\mathcal{A}_1(\mathbb{C})$.

\medskip

In order to show that $F_\delta$ converges as in \eqref{Eq_Fdelta_sinc_convergence} we first estimate the difference between $F_\delta(z)$ and $G_\delta(z)\coloneqq\frac{2}{\delta}e^{-\frac{1}{\delta}}\sinc(\frac{i}{\delta}-az)$. Using the identity $\sinc(\xi)=\frac{1}{2}\int_{-1}^1e^{it\xi}dt$, $\xi\in\mathbb{C}$, we find
\begin{equation*}
F_\delta(z)-G_\delta(z)=\frac{1}{\delta}e^{-\frac{1}{\delta}}\int_{-1}^1\big(e^{itR_\delta(z)}-e^{it(\frac{i}{\delta}-az)}\big)dt,
\end{equation*}
and with the help of \eqref{Eq_Exponential_estimate} we obtain
\begin{equation}\label{Eq_Fdelta_Gdelta_estimate}
\begin{split}
|F_\delta(z)-G_\delta(z)|&\leq\frac{1}{\delta}e^{-\frac{1}{\delta}}\int_{-1}^1|t|\Big|R_\delta(z)+az-\frac{i}{\delta}\Big|e^{|t|\max\{|R_\delta(z)|,|\frac{i}{\delta}-az|\}}dt  \\
&\leq\frac{1}{\delta}e^{-\frac{1}{\delta}}\int_{-1}^1\Big|R_\delta(z)+az-\frac{i}{\delta}\Big|e^{|t|(|R_\delta(z)+az-\frac{i}{\delta}|+|az-\frac{i}{\delta}|)}dt.
\end{split}
\end{equation}
Since the estimate \eqref{Eq_Root_estimate} translates into
\begin{equation*}
\Big|R_\delta(z)+az-\frac{i}{\delta}\Big|=\frac{1}{\delta}\big|\sqrt{\delta^2z^2-2ia\delta z-1}+a\delta z-i\big|\leq C\min\{|z|,\delta|z|^2\},
\end{equation*}
we can further estimate \eqref{Eq_Fdelta_Gdelta_estimate} in the form
\begin{equation}\label{Eq_Fdelta_estimate_1}
\begin{split}
|F_\delta(z)-G_\delta(z)|&\leq C|z|^2e^{-\frac{1}{\delta}}\int_{-1}^1e^{|t|(C|z|+a|z|+\frac{1}{\delta})}dt\\
&\leq C|z|^2e^{(C+a)|z|}e^{-\frac{1}{\delta}}\int_{-1}^1e^{\frac{|t|}{\delta}}dt  \\
&=2C\delta|z|^2e^{(C+a)|z|}e^{-\frac{1}{\delta}}(e^{\frac{1}{\delta}}-1)\\
& \leq 2C\delta|z|^2e^{(C+a)|z|};
\end{split}
\end{equation}
in the last inequality we used $|z|^2\leq 2e^{|z|}$, which follows from the power series representation of the exponential.

\medskip

In a second step we estimate the difference between $G_\delta(z)$ and $e^{iaz}$. For this, we again use $\sinc(\xi)=\frac{1}{2}\int_{-1}^1e^{it\xi}dt$, $\xi\in\mathbb{C}$, as well as integration by parts to rewrite $G_\delta$ as
\begin{align*}
G_\delta(z)&=\frac{1}{\delta}e^{-\frac{1}{\delta}}\int_{-1}^1e^{it(\frac{i}{\delta}-az)}dt\\
&=-e^{-\frac{1}{\delta}}\int_{-1}^1\Big(\frac{d}{dt}e^{-\frac{t}{\delta}}\Big)e^{-iatz}dt \\
&=e^{iaz}-e^{-\frac{2}{\delta}}e^{-iaz}-iaze^{-\frac{1}{\delta}}\int_{-1}^1e^{-\frac{t}{\delta}}e^{-iatz}dt.
\end{align*}
This representation allows the estimate
\begin{equation}\label{Eq_Fdelta_estimate_2}
\begin{split}
|G_\delta(z)-e^{iaz}|&\leq\Big(e^{-\frac{2}{\delta}}+a|z|e^{-\frac{1}{\delta}}\int_{-1}^1e^{-\frac{t}{\delta}}dt\Big)e^{a|z|}\\
&=\big(e^{-\frac{2}{\delta}}+a\delta|z|(1-e^{-\frac{2}{\delta}})\big)e^{a|z|} \\
&\leq(e^{-\frac{2}{\delta}}+a\delta|z|)e^{a|z|}\\
&\leq(e^{-\frac{2}{\delta}}+a\delta)e^{(a+1)|z|}, 
\end{split}
\end{equation}
where in the last inequality we used $|z|\leq e^{|z|}$, which follows from the power series representation of the exponential.

\medskip

Combining now \eqref{Eq_Fdelta_estimate_1} and \eqref{Eq_Fdelta_estimate_2} leads to the estimate
\begin{equation*}
|F_\delta(z)-e^{iaz}|e^{-(C+a+1)|z|}\leq 4C\delta+e^{-\frac{2}{\delta}}+a\delta,\quad z\in\mathbb{C},
\end{equation*}
and since the right hand side is $z$-independent and converges to zero as $\delta\rightarrow 0^+$, this shows 
$\lim_{\delta\rightarrow 0^+}F_\delta(z)=e^{iaz}$ in $\mathcal{A}_1(\mathbb{C})$.
\end{proof}

The convergence result of Theorem~\ref{satz_Fdelta_sinc} together with the integral representation of the 
$\sinc$-function 
in Lemma~\ref{lem_Integral_representation_of_sinc} imply the superoscillatory property 
of the functions $F_\delta$ in \eqref{Eq_Fdelta_sinc}.

\begin{cor}
For every $a>1$ the functions $F_\delta$ in \eqref{Eq_Fdelta_sinc} are superoscillating with limit\begin{equation}\label{Eq_Fdelta_cor_convergence}
\lim\limits_{\delta\rightarrow 0^+}F_\delta(z)=e^{iaz}\quad\text{in }\mathcal{A}_1(\mathbb{C}).
\end{equation}
\end{cor}

\begin{proof}
The convergence \eqref{Eq_Fdelta_cor_convergence} was shown in Theorem~\ref{satz_Fdelta_sinc} and it remains the integral representation \eqref{Eq_Superoscillating_integral}. From \eqref{Eq_Integral_representation_of_sinc} with $z$ and $b$ replaced by $z-\frac{ia}{\delta}$ and 
$\frac{\sqrt{a^2-1}}{\delta}$, respectively, we obtain
\begin{equation*}
F_\delta(z)=\frac{2}{\delta}e^{-\frac{1}{\delta}}\sinc\bigg(\sqrt{z^2-\frac{2iaz}{\delta}-\frac{1}{\delta^2}}\bigg)=\frac{1}{\delta}\int_{-1}^1e^{ikz}e^{\frac{ak-1}{\delta}}J_0\bigg(\frac{\sqrt{(a^2-1)(1-k^2)}}{\delta}\bigg)dk.
\end{equation*}
Using the complex Borel measure
\begin{equation*}
\mu_\delta(B)\coloneqq\frac{1}{\delta}\int_Be^{\frac{ak-1}{\delta}}J_0\bigg(\frac{\sqrt{(a^2-1)(1-k^2)}}{\delta}\bigg)dk
\end{equation*}
for Borel sets $B\subseteq[-1,1]$ we conclude that $F_\delta$ admits the integral representation \eqref{Eq_Superoscillating_integral} with $k_0=1$ and hence the functions $F_\delta$ are superoscillating according to Definition~\ref{defi_Superoscillation}.
\end{proof}

\section{Superoscillations with prescribed frequencies}\label{sec_Superoscillations_with_prescribed_frequencies}

The standard example \eqref{Eq_Example} of a superoscillating function is a linear combination of plane waves
\begin{equation}\label{Eq_Constructing_example}
F_n(z)=\sum\limits_{j=0}^nC_j(n)e^{ik_j(n)z},\quad z\in\mathbb{C},
\end{equation}
with frequencies $k_j(n)=1-\frac{2j}{n}$ and coefficients $C_j(n)=\binom{n}{j}(\frac{1+k}{2})^{n-j}(\frac{1-k}{2})^j$. For a long time it was not clear how many of such superoscillating functions exist. 
In the recent paper \cite{ACSSST21} for any set of given frequencies $k_j(n)\in[-1,1]$ and any target frequency $a\in\mathbb{R}\setminus[-1,1]$ coefficients $C_j(n)$ were constructed, such that the corresponding sequence \eqref{Eq_Constructing_example} approximates the exponential $e^{iax}$ in the sense that the Taylor series 
of $F_n(x)$ and $e^{iax}$ coincide up to order $n$. This approximation suggests a certain 
superoscillatory behaviour, but is not sufficient to conclude the convergence $F_n(z)\rightarrow e^{iaz}$ 
in $\mathcal{A}_1(\mathbb{C})$. 
The analysis in this section is partly inspired by these earlier considerations in \cite{ACSSST21}. 
In particular, in Corollary~\ref{cor_Fn_exponentials} we conclude the $\mathcal{A}_1$-convergence
under sufficient conditions on the frequencies $k_j(n)$.

\medskip

Starting from a more general perspective we shall study functions of the form
\begin{equation}\label{Eq_Fn2}
F_n(z)=\sum\limits_{j=0}^nC_j(n)f_{j,n}(z),\quad z\in\mathbb{C},
\end{equation}
with frequency functions
\begin{equation}\label{Eq_fjn2}
f_{j,n}(z)=\int_{-k_0}^{k_0}e^{ikz}d\mu_{j,n}(k),\quad z\in\mathbb{C}.
\end{equation}
The idea is now to choose the coefficients $C_j(n)$ in such a way that the Taylor coefficients of $F_n(z)$ and $e^{iaz}$ coincide up to order $n$, that is,
\begin{equation}\label{Eq_Fn_assumption}
F_n^{(l)}(0)=(ia)^l,\quad n\in\mathbb{N},\,l\in\{0,\dots,n\},
\end{equation}
or, more explicitly using \eqref{Eq_fjn2} and \eqref{Eq_Fn2} 
\begin{equation*}
\sum\limits_{j=0}^nC_j(n)\int_{-k_0}^{k_0}(ik)^ld\mu_{j,n}(k)=(ia)^l,\quad n\in\mathbb{N},\,l\in\{0,\dots,n\}.
\end{equation*}
For each $n\in\mathbb{N}$ this is a linear system of the form
\begin{equation}\label{Eq_Sn_system}
S_n\mathbf{c}_n=\mathbf{a}_n,\quad n\in\mathbb{N},
\end{equation}
where the matrix $S_n\in\mathbb{C}^{(n+1)\times(n+1)}$ and the vectors 
$\mathbf{c}_n,\mathbf{a}_n\in\mathbb{C}^{n+1}$ are given by
\begin{equation}\label{Eq_Sn}
S_n=\begin{pmatrix} \int_{-k_0}^{k_0}k^0d\mu_{0,n}(k) & \dots & \int_{-k_0}^{k_0}k^0d\mu_{n,n}(k) \\ \vdots & \ddots & \vdots \\ \int_{-k_0}^{k_0}k^nd\mu_{0,n}(k) & \dots & \int_{-k_0}^{k_0}k^nd\mu_{n,n}(k) \end{pmatrix},\quad\mathbf{c}_n=\begin{pmatrix} C_0(n) \\ \vdots \\ C_n(n) \end{pmatrix},\quad\mathbf{a}_n=\begin{pmatrix} a^0 \\ \vdots \\ a^n \end{pmatrix}.
\end{equation}
The following Theorem~\ref{satz_Fn_convergence} is of abstract nature and provides conditions such that 
the functions $F_n$ in \eqref{Eq_Fn2} are superoscillating. Using point measures $\mu_{j,n}$ in Corollary~\ref{cor_Fn_exponentials} we complete the earlier considerations in \cite{ACSSST21}.
A different situation based on the choice of absolutely continuous measures is treated in Corollary~\ref{cor_Fn_derivative_convergence}.

\begin{satz}\label{satz_Fn_convergence}
Let $k_0>0$, $a\in\mathbb{R}\setminus[-k_0,k_0]$, $\mu_{j,n}$ be complex Borel measures on $[-k_0,k_0]$, 
and let $\mathbf{c}_n$ be a solution of the system \eqref{Eq_Sn_system} for $n\in\mathbb{N}$. 
If there exist constants $\kappa_1,\kappa_2\geq 0$ such that for every $n\in\mathbb{N}$ and $j\in\{0,\dots,n\}$
\begin{equation}\label{Eq_mu_C_assumption}
|\mu_{j,n}|([-k_0,k_0])\leq\kappa_1^n\quad\text{and}\quad|C_j(n)|\leq\kappa_2^n
\end{equation}
hold, then the functions $F_n$ in \eqref{Eq_Fn2} are superoscillating with limit
\begin{equation}\label{limlim3}
\lim_{n\rightarrow\infty}F_n(z)=e^{iaz}\quad\text{in }\mathcal{A}_1(\mathbb{C}).
\end{equation}
\end{satz} 

\begin{proof}
First of all, it is clear that $f_{j,n}\in\mathcal{A}_1(\mathbb{C})$ for every $j\in\{0,\dots,n\}$ and hence $F_n\in\mathcal{A}_1(\mathbb{C})$ by the arguments given below Definition \ref{defi_Superoscillation}. The bounds of the measures $\mu_{j,n}$ and the coefficients $C_j(n)$ in \eqref{Eq_mu_C_assumption} immediately lead to the estimate
\begin{equation}\label{Eq_Fn_bound}
|F_n(z)|\leq\sum\limits_{j=0}^n|C_j(n)|\int_{-k_0}^{k_0}e^{|k||z|}d|\mu_{j,n}|(k)\leq(n+1)(\kappa_1\kappa_2)^ne^{k_0|z|},\quad z\in\mathbb{C}.
\end{equation}
Next, we write the difference between $F_n(z)$ and $e^{iaz}$ as the Taylor series
\begin{equation}\label{Eq_Difference_Taylorseries}
F_n(z)-e^{iaz}=\sum\limits_{l=0}^\infty\frac{F_n^{(l)}(0)-(ia)^l}{l!}z^l=\sum\limits_{l=n+1}^\infty\frac{F_n^{(l)}(0)-(ia)^l}{l!}z^l,\quad z\in\mathbb{C},
\end{equation}
where in the second equality we used \eqref{Eq_Fn_assumption}. For $z\in\mathbb{C}\setminus\{0\}$ the Cauchy integral formula along a circle of radius $(1+\kappa_1\kappa_2)|z|$ leads to the following bounds for the coefficients in
\eqref{Eq_Difference_Taylorseries}
\begin{align*}
\frac{|F_n^{(l)}(0)-(ia)^l|}{l!}&=\bigg|\frac{1}{2\pi i}\int_{|\xi|=(1+\kappa_1\kappa_2)|z|}\frac{F_n(\xi)-e^{ia\xi}}{\xi^{l+1}}d\xi\bigg| \\
&\leq\frac{1}{2\pi(1+\kappa_1\kappa_2)^l|z|^l}\int_0^{2\pi}\Big|F_n\big((1+\kappa_1\kappa_2)|z|e^{i\varphi}\big)-e^{ia(1+\kappa_1\kappa_2)|z|e^{i\varphi}}\Big|d\varphi \\
&\leq\frac{1}{2\pi(1+\kappa_1\kappa_2)^l|z|^l}\int_0^{2\pi}\Big((n+1)(\kappa_1\kappa_2)^ne^{k_0(1+\kappa_1\kappa_2)|z|}+e^{|a|(1+\kappa_1\kappa_2)|z|}\Big)d\varphi \\
&\leq\frac{(n+1)(\kappa_1\kappa_2)^n+1}{(1+\kappa_1\kappa_2)^l|z|^l}e^{|a|(1+\kappa_1\kappa_2)|z|},
\end{align*}
where we used \eqref{Eq_Fn_bound} and $k_0<|a|$ in the estimate of the integrand. Plugging this into the Taylor series \eqref{Eq_Difference_Taylorseries} gives
\begin{equation}\label{Eq_Difference_estimate}
\begin{split}
|F_n(z)-e^{iaz}|&\leq\big((n+1)(\kappa_1\kappa_2)^n+1\big)e^{|a|(1+\kappa_1\kappa_2)|z|}\sum\limits_{l=n+1}^\infty\frac{1}{(1+\kappa_1\kappa_2)^l}  \\
&=\frac{(n+1)(\kappa_1\kappa_2)^n+1}{\kappa_1\kappa_2(1+\kappa_1\kappa_2)^n}e^{|a|(1+\kappa_1\kappa_2)|z|},\quad z\in\mathbb{C}\setminus\{0\}. 
\end{split}
\end{equation}
Due to \eqref{Eq_Difference_Taylorseries} this inequality obviously holds for $z=0$ and hence we conclude 
the $\mathcal{A}_1$-convergence
\begin{equation*}
\sup\limits_{z\in\mathbb{C}}|F_n(z)-e^{iaz}|e^{-|a|(1+\kappa_1\kappa_2)|z|}\leq\frac{(n+1)(\kappa_1\kappa_2)^n+1}{\kappa_1\kappa_2(1+\kappa_1\kappa_2)^n}\rightarrow 0,\quad\text{as }n\rightarrow\infty.
\end{equation*}
Finally, the integral representation \eqref{Eq_Superoscillating_integral} of the functions $F_n$ is satisfied by
\begin{equation*}
F_n(z)=\int_{-k_0}^{k_0}e^{ikz}\sum\limits_{j=0}^nC_j(n)\mu_{j,n}(k)=\int_{-k_0}^{k_0}e^{ikz}d\mu_n(k),\quad z\in\mathbb{C},
\end{equation*}
using the Borel measure $\mu_n(B)\coloneqq\sum_{j=0}^nC_j(n)\mu_{j,n}(B)$ for Borel sets $B\subseteq [-k_0,k_0]$. 
We have shown that the functions $F_n$ are superoscillating with limit \eqref{limlim3}.
\end{proof}

In the next corollary we return to the initial problem of this section. The aim is to find for given  $k_j(n)$ coefficients $C_j(n)$ in \eqref{Eq_Constructing_example} such that the resulting functions $F_n$ are superoscillating. Here the measures $\mu_{j,n}$ in Theorem~\ref{satz_Fn_convergence} 
are chosen as point measures and an additional condition on the frequencies $k_j(n)$ is imposed.

\begin{cor}\label{cor_Fn_exponentials}
Let $k_0>0$, $k_j(n)\in[-k_0,k_0]$, and assume that 
\begin{equation}\label{Eq_Frequency_distance}
\prod\limits_{l=0,l\neq j}^n|k_l(n)-k_j(n)|\geq\kappa^n,\quad n\in\mathbb{N},\,j\in\{0,\dots,n\},
\end{equation}
holds for some $\kappa>0$. For $a\in\mathbb{R}\setminus[-k_0,k_0]$ define the coefficients
\begin{equation}\label{Eq_Cjn}
C_j(n)=\prod\limits_{l=0,l\neq j}^n\frac{k_l(n)-a}{k_l(n)-k_j(n)},\quad n\in\mathbb{N},\,j\in\{0,\dots,n\}.
\end{equation}
Then the functions
\begin{equation}\label{Eq_Fn_exponentials}
F_n(z)\coloneqq\sum_{j=0}^nC_j(n)e^{ik_j(n)z},\quad z\in\mathbb{C},
\end{equation}
are superoscillating with limit $
\lim_{n\rightarrow\infty}F_n(z)=e^{iaz}$ in $\mathcal{A}_1(\mathbb{C})$.
\end{cor}

\begin{proof}
Observe first that the functions \eqref{Eq_Fn_exponentials} can be written in the form \eqref{Eq_Fn2} 
using the point measures
\begin{equation*}
\mu_{j,n}(B)\coloneqq\begin{cases} 1, & \text{if }k_j(n)\in B, \\ 0, & \text{if }k_j(n)\notin B, \end{cases}
\end{equation*}
for Borel sets $B\subseteq[-k_0,k_0]$.
Since the first bound in \eqref{Eq_mu_C_assumption} is trivially satisfied, it remains to check that
the (unique) solution $C_j(n)$ of \eqref{Eq_Sn_system} satisfies the second bound
in \eqref{Eq_mu_C_assumption}. The explicit form of the matrix $S_n$ in this particular setting is
\begin{equation*}
S_n=\begin{pmatrix} k_0(n)^0 & \dots & k_n(n)^0 \\ \vdots & \ddots & \vdots \\ k_0(n)^n & \dots & k_n(n)^n \end{pmatrix}.
\end{equation*}
This is a Vandermonde matrix, which is known to be invertible whenever the $(k_j(n))_{j=0}^n$ are pairwise disjoint. But this is obviously the case due to the assumption \eqref{Eq_Frequency_distance}. Moreover, using Cramer's rule and the known determinant formula for Vandermonde matrices, one can easily derive the explicit representation
\eqref{Eq_Cjn} for the coefficients $C_j(n)$
of the solution, see also \cite[Theorem 2.2]{ACSSST21}. The assumption \eqref{Eq_Frequency_distance} now allows to estimate these coefficients as
\begin{equation*}
|C_j(n)|=\prod\limits_{l=0,l\neq j}^n\frac{|k_l(n)-a|}{|k_l(n)-k_j(n)|}\leq\prod\limits_{l=0,l\neq j}^n\frac{k_0+|a|}{|k_l(n)-k_j(n)|}\leq\Big(\frac{k_0+|a|}{\kappa}\Big)^n,
\end{equation*}
which shows the second bound in \eqref{Eq_mu_C_assumption}. It follows from 
Theorem~\ref{satz_Fn_convergence} that the functions $F_n$ are superoscillating.
\end{proof}

\begin{bem}
Note that the frequencies $k_j(n)=1-\frac{2j}{n}$ in \eqref{Eq_Example_coefficients} can be estimated by
\begin{equation*}
\prod\limits_{l=0,l\neq j}^n|k_l(n)-k_j(n)|=\prod\limits_{l=0,l\neq j}^n\frac{2|l-j|}{n}=\frac{2^nj!(n-j)!}{n^n}=\frac{2^nn!}{n^n{n\choose j}}\geq\frac{1}{e^n},
\end{equation*}
where in the last inequality we used ${n\choose j}\leq 2^n$ as well as $\frac{n^n}{n!}\leq e^n$, which is a consequence of the Stirling formula. Hence, Corollary~\ref{cor_Fn_exponentials} applies and the coefficients 
$C_j(n)$ in \eqref{Eq_Cjn} lead to superoscillating functions $F_n$ in \eqref{Eq_Fn_exponentials}. 
Note that these coefficients do not coincide with those in \eqref{Eq_Example_coefficients}.
\end{bem}

Now we turn to a different situation, where the measures $\mu_{j,n}$ in Theorem~\ref{satz_Fn_convergence} are absolutely continuous. Here the initial idea is to replace the exponentials $e^{ik_j(n)z}$ in \eqref{Eq_Constructing_example} by the derivatives of a bandlimited function
\begin{equation}\label{Eq_f_derivative}
f(z)=\int_{-k_0}^{k_0}e^{ikz}h(k)dk,\quad z\in\mathbb{C},
\end{equation}
with $h\in L^1([-k_0,k_0])$;
note that $h$ is the compactly supported Fourier transform of $f$.

\begin{cor}\label{cor_Fn_derivative_convergence}
Let $k_0>0$ and let $h\in L^1([-k_0,k_0])$ 
be a nonnegative function from the Szeg\"o class, i.e. for some $\alpha,\beta\in[-k_0,k_0]$ with $\alpha<\beta$ one has
\begin{equation}\label{Eq_Szego_condition}
\int_\alpha^\beta\frac{\ln(h(k))}{\sqrt{(\beta-k)(k-\alpha)}}dk>-\infty,
\end{equation}
and let $f$ be as in \eqref{Eq_f_derivative}.
Then, for every $a\in\mathbb{R}\setminus[-k_0,k_0]$ there exist coefficients $C_j(n)\in\mathbb{C}$ such that the functions 
\begin{equation}\label{jaja}
F_n(z)=\sum_{j=0}^n C_j(n) f^{(j)}(z),\quad z\in\mathbb{C},
\end{equation}
are superoscillating with limit $
\lim_{n\rightarrow\infty}F_n(z)=e^{iaz}$ in $\mathcal{A}_1(\mathbb{C})$.
\end{cor}

\begin{proof}
Note that the functions $F_n$ in \eqref{jaja} are of the form \eqref{Eq_Fn2} if we use the measures
\begin{equation}\label{Eq_Moment_measure}
\mu_{j,n}(B)=\int_B(ik)^jh(k)dk
\end{equation}
for Borel sets $B\subseteq[-k_0,k_0]$.
Hence we are in the situation to apply Theorem~\ref{satz_Fn_convergence} but have to check the two estimates \eqref{Eq_mu_C_assumption}. The bound of the absolute variations of the measures is clearly satisfied by
\begin{equation*}
|\mu_{j,n}|([-k_0,k_0])=\int_{-k_0}^{k_0}|k|^j|h(k)|dk\leq k_0^j\int_{-k_0}^{k_0}|h(k)|dk\leq\max\{k_0,1\}^n\Vert h\Vert_{L^1}.
\end{equation*}
To calculate the coefficients $C_j(n)$ we have to solve the system \eqref{Eq_Sn_system} or, equivalently, 
the system $M_n\mathbf{d}_n=\mathbf{a}_n$ with the matrix and the vectors
\begin{equation*}
(M_n)_{j,l}=\int_{-k_0}^{k_0}k^{j+l}h(k)dk,\quad(\mathbf{d}_n)_j=i^jC_j(n),\quad(\mathbf{a}_n)_l=a^l,
\end{equation*}
where we shifted the powers $i^j$ of the imaginary unit from the coefficient matrix into the solution vector. Now consider the matrix
\begin{equation*}
(\widetilde{M}_n)_{j,l}\coloneqq\int_\alpha^\beta k^{j+l}h(k)dk,\quad j,l\in\{0,\dots,n\}.
\end{equation*}
It is shown in \cite{S36} that its lowest eigenvalue $\widetilde{\lambda}_n=\inf_{0\neq x\in\mathbb{C}^n}\frac{\langle\widetilde{M}_nx,x\rangle}{|x|^2}$ is bounded from below by
\begin{equation*}
\widetilde{\lambda}_n\geq\kappa^n,\quad n\in\mathbb{N},
\end{equation*}
for some $\kappa>0$. Hence $\widetilde{M}_n$ is elliptic with bound $\widetilde{M}_n\geq\kappa^n$ for every $n\in\mathbb{N}$. For $x\in\mathbb{C}^{n+1}$ the estimate
\begin{equation*}
\langle M_nx,x\rangle=\int_{-k_0}^{k_0}\Big|\sum\limits_{j=0}^nk^jx_j\Big|^2h(k)dk\geq\int_\alpha^\beta\Big|\sum\limits_{j=0}^nk^jx_j\Big|^2h(k)dk=\langle\widetilde{M}_nx,x\rangle\geq\kappa^n|x|^2,
\end{equation*}
shows that the same is true for the matrix $M_n$. Consequently, the inverse matrix exists and is bounded by $\Vert M_n^{-1}\Vert\leq\frac{1}{\kappa^n}$. Finally we can estimate the elements $C_j(n)$ of the solution by
\begin{equation*}
|C_j(n)|\leq|\mathbf{d}_n|=|M_n^{-1}\mathbf{a}_n|\leq\frac{1}{\kappa^n}|\mathbf{a}_n|\leq\frac{\sqrt{n+1}\max\{a^n,1\}}{\kappa^n}.
\end{equation*}
Hence, also the second bound in \eqref{Eq_mu_C_assumption} is satisfied and the assertions follow from
Theorem~\ref{satz_Fn_convergence}.
\end{proof}

In the following example we make a special choice of the function $h$ in \eqref{Eq_f_derivative} in order to find a close connection between the functions \eqref{jaja} and the method presented in the paragraph below equation \eqref{Eq_Ferreira_Kempf_Lee}.

\begin{bsp}\label{bsp_Ferreira_Kempf_Lee}
In this example we apply Corollary \ref{cor_Fn_derivative_convergence} with $k_0=1$ and the constant function $h(k)=\frac{1}{2}$. The exact integral
\begin{equation*}
\int_{-1}^1\frac{\ln(\frac{1}{2})}{\sqrt{(1-k)(k+1)}}dk=-\pi\ln(2)>-\infty
\end{equation*}
shows that the condition \eqref{Eq_Szego_condition} is satisfied in this setting. Moreover, the function $f$ in \eqref{Eq_f_derivative} is given by
\begin{equation*}
f(z)=\frac{1}{2}\int_{-1}^1e^{ikz}dk=\sinc(z),\quad z\in\mathbb{C},
\end{equation*}
which means that the functions $F_n$ in \eqref{jaja} become the sum of derivatives of the $\sinc$-function
\begin{equation}\label{Eq_Sum_of_sinc}
F_n(z)=\sum\limits_{j=0}^nC_j(n)\sinc^{(j)}(z),\quad z\in\mathbb{C}.
\end{equation}
This example is of particular importance, since these functions are closely related to the method of P. Ferreira, A. Kempf and D. Lee
in the series of papers \cite{FK02,FK06,FKR07,FL14_1,FL14_2}.
There functions of the form \eqref{Eq_Ferreira_Kempf_Lee} were considered, with coefficients $(c_l)_l$ chosen in such a way that at prescribed points $(x_l)_l$ the function $F$ admits the prescribed values $F(x_l)=a_l$, $l\in\{0,\dots,m\}$. Furthermore, in the paper \cite{FL14} this construction is extended in the sense that also the values of the derivative $F'$ can be prescribed and it is then straight forward to also prescribe the derivatives up to any order $n$, i.e. $F^{(j)}(x_{j,l})=a_{j,l}$, $l\in\{0,\dots,m\}$, $j\in\{0,\dots,n\}$. This then leads to a function $F$ of the form
\begin{equation*}
F(x)=\sum\limits_{l=0}^mc_{0,l}\sinc(x-x_{0,l})+\sum\limits_{l=0}^mc_{1,l}\sinc'(x-x_{1,l})+\dots+\sum\limits_{l=0}^mc_{n,l}\sinc^{(n)}(x-x_{n,l}).
\end{equation*}
In the special case that we choose $m=0$, the points $x_{0,0}=x_{1,0}=\dots=x_{n,0}=0$ as the origin and the values $a_{j,0}=(ia)^j$, we end up with exactly the function \eqref{Eq_Sum_of_sinc}, namely
\begin{equation*}
F(x)=c_{0,0}\sinc(x)+c_{1,0}\sinc'(x)+\dots+c_{n,0}\sinc^{(n)}(x).
\end{equation*}
\end{bsp}

\section{Superoscillating solutions of the generalized Schr\"odinger equation}\label{sec_Superoscillating_solutions_of_the_generalized_Schroedinger_equation}

In this final section we derive two methods to construct new families of superoscillating functions out of given ones. The underlying technique has its origin in \cite{ACSST15} and was revisited in \cite{ACSST16,ACSST17_1,ASTY18}; it is based on the generalized force free Schr\"odinger equation
\begin{equation}\label{Eq_Generalized_Schroedinger_equation}
i\frac{\partial}{\partial t}\Psi(t,z)=-H\Big(-i\frac{\partial}{\partial z}\Big)\Psi(t,z),\quad t,z\in\mathbb{C}, 
\end{equation}
where for some entire $H(z)=\sum_{l=0}^\infty h_lz^l$ the operator $H(-i\frac{d}{dz})\coloneqq\sum_{l=0}^\infty h_l(-i)^l\frac{d^l}{dz^l}$ is defined as the corresponding infinite order differential expression. If we now consider superoscillating initial conditions $\Psi_n(0,z)=F_n(z)$ of the form \eqref{Eq_Superoscillating_integral}, then
the equation \eqref{Eq_Generalized_Schroedinger_equation} admits the explicit solution
\begin{equation}\label{Eq_Generalized_Schroedinger_solution}
\Psi_n(t,z)=\int_{-k_0}^{k_0}e^{iH(k)t}e^{ikz}d\mu_n(k),\quad t,z\in\mathbb{C};
\end{equation}
cf. Lemma~\ref{lem_U_continuity}.
One may expect that the solutions $\Psi_n(t,\,\cdot\,)$ are superoscillatory also for $t\not=0$. This method is specified in Theorem~\ref{satz_Constructing_Fn}~(i). Moreover, it turns out that in the point $z=0$ the functions $\Psi_n(\,\cdot\,,0)$ are superoscillatory in the time variable $t$ as well, see Theorem~\ref{satz_Constructing_Fn}~(ii).

\medskip

In the next lemma we collect some properties of the propagator of the generalized Schr\"odinger equation \eqref{Eq_Generalized_Schroedinger_equation}.

\begin{lem}\label{lem_U_continuity}
Let $H:\mathbb{C}\rightarrow\mathbb{C}$ be an entire function. For every $t\in\mathbb{C}$ the operator
\begin{equation}\label{Eq_U_H}
U_t\coloneqq\sum\limits_{m=0}^\infty\frac{(it)^m}{m!}H\Big(-i\frac{d}{dz}\Big)^m,
\end{equation}
is continuous as an operator $U_t:\mathcal{A}_1(\mathbb{C})\rightarrow\mathcal{A}_1(\mathbb{C})$ and it acts on plane waves $e^{iaz}$ and functions $F_n(z)$ of the form \eqref{Eq_Superoscillating_integral} as
\begin{equation}\label{Eq_Ut_action}
U_te^{iaz}=e^{iH(a)t}e^{iaz}\quad\text{and}\quad U_tF_n(z)=\int_{-k_0}^{k_0}e^{iH(k)t}e^{ikz}d\mu_n(k),
\end{equation}
respectively. Moreover, for every $F\in\mathcal{A}_1(\mathbb{C})$, the function $\Psi(t,z)=U_tF(z)$, $t,z\in\mathbb{C}$, is a solution of the generalized  Schr\"odinger equation \eqref{Eq_Generalized_Schroedinger_equation} with initial condition $\Psi(0,z)=F(z)$.
\end{lem}

\begin{proof}
The fact that the operator $U_t$ is continuous in $\mathcal{A}_1(\mathbb{C})$ was already shown 
in \cite[Theorem~2.7]{ACPS21}. To investigate the action \eqref{Eq_Ut_action} of $U_t$ on plane waves $e^{iaz}$, we start with the power series representation $H(z)=\sum_{l=0}^\infty h_lz^l$. Then the operator $H(-i\frac{d}{dz})$ acts on plane waves $e^{iaz}$ as multiplication
\begin{equation*}
H\Big(-i\frac{d}{dz}\Big)e^{iaz}=\sum\limits_{l=0}^\infty h_l(-i)^l\frac{d^l}{dz^l}e^{iaz}=\sum\limits_{l=0}^\infty h_la^le^{iaz}=H(a)e^{iaz}
\end{equation*}
and hence the operator $U_t$ acts as
\begin{equation}\label{Eq_Ut_action_on_exponentials}
U_te^{iaz}=\sum\limits_{m=0}^\infty\frac{(it)^m}{m!}H\Big(-i\frac{d}{dz}\Big)^me^{iaz}=\sum\limits_{m=0}^\infty\frac{(iH(a)t)^m}{m!}e^{iaz}=e^{iH(a)t}e^{iaz}.
\end{equation}
To determine the action of $U_t$ on functions $F_n(z)=\int_{-k_0}^{k_0}e^{ikz}d\mu_n(k)$ we 
have to interchange the complex derivatives with the integral. In the first step we use a version of the dominated convergence theorem and obtain
\begin{equation*}
\frac{d^l}{dz^l}\int_{-k_0}^{k_0}e^{ikz}d\mu_n(k)=\int_{-k_0}^{k_0}\frac{d^l}{dz^l}e^{ikz}d\mu_n(k).
\end{equation*}
Applying the dominated convergence theorem once more, we are also allowed to carry the infinite sum $\sum_{l=0}^\infty$ of the operator $H(-i\frac{d}{dz})$ inside the integral. This gives
\begin{equation*}
H\Big(-i\frac{d}{dz}\Big)\int_{-k_0}^{k_0}e^{ikz}d\mu_n(k)=\int_{-k_0}^{k_0}H\Big(-i\frac{d}{dz}\Big)e^{ikz}d\mu_n(k).
\end{equation*}
Using the dominated convergence a third time for the series $\sum_{m=0}^\infty$ of the operator $U_t$ finally gives
\begin{equation*}
U_t\int_{-k_0}^{k_0}e^{ikz}d\mu_n(k)=\int_{-k_0}^{k_0}U_te^{ikz}d\mu_n(k)=\int_{-k_0}^{k_0}e^{iH(k)t}e^{ikz}d\mu_n(k),
\end{equation*}
where in the second equation we used \eqref{Eq_Ut_action_on_exponentials}.

\medskip

In the last part of the proof let $F\in\mathcal{A}_1(\mathbb{C})$ and consider $\Psi(t,z)\coloneqq U_tF(z)$, $t,z\in\mathbb{C}$. Interchanging derivatives and sums for similar reasons as above, we obtain that
\begin{align*}
i\frac{\partial}{\partial t}\Psi(t,z)&=i\frac{\partial}{\partial t}\sum\limits_{m=0}^\infty\frac{(it)^m}{m!}H\Big(-i\frac{d}{dz}\Big)^mF(z) \\
&=-\sum\limits_{m=1}^\infty\frac{(it)^{m-1}}{(m-1)!}H\Big(-i\frac{d}{dz}\Big)^mF(z) \\
&=-H\Big(-i\frac{d}{dz}\Big)\sum\limits_{m=0}^\infty\frac{(it)^m}{m!}H\Big(-i\frac{d}{dz}\Big)^mF(z) \\
&=-H\Big(-i\frac{d}{dz}\Big)\Psi(t,z).
\end{align*}
Hence $\Psi(t,z)$ in indeed a solution of the generalized Schr\"odinger equation \eqref{Eq_Generalized_Schroedinger_equation}. Since $U_0=\text{id}$ reduces to the identity operator, the solution also satisfies the initial condition $\Psi(0,z)=F(z)$.
\end{proof}

The novelty of Theorem~\ref{satz_Constructing_Fn} below is that 
arbitrary superoscillating functions of the form \eqref{Eq_Superoscillating_integral} are allowed, while in \cite{ACSST15,ACSST16,ACSST17_1,ASTY18} only the standard example \eqref{Eq_Example} was considered. Furthermore, 
it is shown that the resulting functions $F_n^{(1)}$ and $F_n^{(2)}$ in \eqref{Eq_Constructing_Fn1} and \eqref{Eq_Constructing_Fn2} converge in $\mathcal{A}_1(\mathbb{C})$, while only uniform convergence on compact sets was proven in \cite{ACSST15,ACSST16,ACSST17_1,ASTY18}.

\begin{satz}\label{satz_Constructing_Fn}
Let $k_0>0$, $a\in\mathbb{R}\setminus[-k_0,k_0]$, and $\mu_n$ be complex Borel measures on $[-k_0,k_0]$.
Assume that the functions 
\begin{equation}\label{Eq_Constructing_Fn}
F_n(z)=\int_{-k_0}^{k_0}e^{ikz}d\mu_n(k),\quad z\in\mathbb{C},
\end{equation}
are superoscillating with limit 
$\lim_{n\rightarrow\infty}F_n(z)=e^{iaz}$ in $\mathcal{A}_1(\mathbb{C})$.
Then one can construct the following two new families of superoscillating functions:
\begin{enumerate}
\item[{\rm(i)}] For every entire function $H:\mathbb{C}\rightarrow\mathbb{C}$ the functions
\begin{equation}\label{Eq_Constructing_Fn1}
F_n^{(1)}(z)\coloneqq e^{-H(a)}\int_{-k_0}^{k_0}e^{H(k)}e^{ikz}d\mu_n(k),\quad z\in\mathbb{C},
\end{equation}
are superoscillating with $\lim_{n\rightarrow\infty}F_n^{(1)}(z)=e^{iaz}$ in $\mathcal{A}_1(\mathbb{C})$.

\item[{\rm(ii)}] For every entire function $H:\mathbb{C}\rightarrow\mathbb{C}$ with satisfies $\ran\big(H|_{[-k_0,k_0]}\big)\subseteq[-h_0,h_0]$ and $H(a)\in\mathbb{R}\setminus[-h_0,h_0]$ for some $h_0>0$ the sequence
\begin{equation}\label{Eq_Constructing_Fn2}
F_n^{(2)}(z)\coloneqq\int_{-k_0}^{k_0}e^{iH(k)z}d\mu_n(k),\quad z\in\mathbb{C},
\end{equation}
is superoscillating with $\lim_{n\rightarrow\infty}F_n^{(2)}(z)=e^{iH(a)z}$ in $\mathcal{A}_1(\mathbb{C})$.
\end{enumerate}
\end{satz}

\begin{proof}
For the proof of (i) we consider the operator \eqref{Eq_U_H} with $t=-i$. From the action of this operator in \eqref{Eq_Ut_action} we obtain
\begin{equation*}
U_{-i}F_n(z)=\int_{-k_0}^{k_0}e^{H(k)}e^{ikz}d\mu_n(k)=e^{H(a)}F_n^{(1)}(z),\quad z\in\mathbb{C}.
\end{equation*}
Since $U_{-i}$ is continuous in $\mathcal{A}_1(\mathbb{C})$ due to Lemma~\ref{lem_U_continuity} we conclude the convergence
\begin{equation}\label{Eq_Fn1_convergence}
\lim\limits_{n\rightarrow\infty}F_n^{(1)}(z)=e^{-H(a)}\lim\limits_{n\rightarrow\infty}U_{-i}F_n(z)=e^{-H(a)}U_{-i}e^{iaz}=e^{iaz}\quad\text{in }\mathcal{A}_1(\mathbb{C}),
\end{equation}
where the first identity in \eqref{Eq_Ut_action} was used in the last step.
The fact that $F_n^{(1)}$ is of the form \eqref{Eq_Superoscillating_integral} is clear using the complex Borel measure
\begin{equation*}
\sigma_n(B)=e^{-H(a)}\int_BH(k)d\mu_n(k)
\end{equation*}
for Borel sets $B\subseteq[-k_0,k_0]$.
Hence, the functions $F_n^{(1)}$ are superoscillating with $\lim_{n\rightarrow\infty}F_n^{(1)}(z)=e^{iaz}$ in $\mathcal{A}_1(\mathbb{C})$.

\medskip

(ii)\;\;According to the $\mathcal{A}_1$-convergence \eqref{Eq_Superoscillating_convergence}, there exists some $B\geq 0$, such that
\begin{equation*}
A_n\coloneqq\sup\limits_{z\in\mathbb{C}}|F_n(z)-e^{iaz}|e^{-B|z|}\rightarrow 0,\quad\text{as }n\rightarrow\infty.
\end{equation*}
By the choice of the constants $A_n$, we obtain the estimate
\begin{equation*}
|F_n(z)-e^{iaz}|\leq A_ne^{B|z|},\quad z\in\mathbb{C}.
\end{equation*}
In the proof of \cite[Theorem 2.7]{ACPS21} the authors obtain the estimate
\begin{equation*}
|U_t(F_n(z)-e^{iaz})|\leq A_ne^{\widetilde{D}|t|}e^{\widetilde{B}|z|},\quad t,z\in\mathbb{C},
\end{equation*}
for some $\widetilde{B},\widetilde{D}\geq 0$. For $z=0$ this estimate shows the convergence
\begin{equation*}
\lim\limits_{n\rightarrow\infty}\sup\limits_{t\in\mathbb{C}}\big|U_t(F_n(z)-e^{iaz})|_{z=0}\big|e^{-\widetilde{D}|t|}\leq\lim\limits_{n\rightarrow\infty}A_n=0.
\end{equation*}
In other words, we conclude the $\mathcal{A}_1$-convergence
\begin{equation*}
\lim\limits_{n\rightarrow\infty}U_tF_n(z)|_{z=0}=U_te^{iaz}|_{z=0}
\end{equation*}
in the variable $t$. Using the identities \eqref{Eq_Ut_action} this convergence can be written as
\begin{equation}\label{Eq_Fn2_convergence}
\lim\limits_{n\rightarrow\infty}F_n^{(2)}(t)=\lim\limits_{n\rightarrow\infty}\int_{-k_0}^{k_0}e^{iH(k)t}d\mu_n(k)=\lim\limits_{n\rightarrow\infty}U_tF_n(z)|_{z=0}=U_te^{iaz}|_{z=0}=e^{iH(a)t}.
\end{equation}
Since $\ran\big(H|_{[-k_0,k_0]}\big)\subseteq[-h_0,h_0]$ by assumption, we can choose the complex Borel measures
\begin{equation*}
\sigma_n(B)\coloneqq\mu_n\big(\Set{k\in[-k_0,k_0] | H(k)\in B}\big)
\end{equation*}
for Borel sets $B\subseteq[-h_0,h_0]$ and 
transform $F_n^{(2)}$ into the form 
\begin{equation}\label{Eq_Constructing_Fnhhh}
F_n^{(2)}(t)=\int_{-h_0}^{h_0}e^{ikt}d\sigma_n(k),\quad z\in\mathbb{C},
\end{equation}
which is the 
representation \eqref{Eq_Superoscillating_integral}. Together with the convergence \eqref{Eq_Fn2_convergence} to a plane wave with frequency $H(a)\in\mathbb{R}\setminus[-h_0,h_0]$ this shows that the functions $F_n^{(2)}$ are superoscillating.
\end{proof}

\appendix

\section{}

This appendix contains two technical results which are used in Section~\ref{sec_Berry_superoscillations} and Section~\ref{sec_Superoscillating_sinc_function}. In Lemma~\ref{Eq_Exponential_estimate} we prove an elementary estimate for the difference of two complex exponentials and in Lemma~\ref{lem_Integral_representation_of_sinc} we derive an integral representation of the $\sinc$-function.

\begin{lem}\label{lem_Exponential_estimate}
For every $z_1,z_2\in\mathbb{C}$ one has
\begin{equation}\label{Eq_Exponential_estimate}
|e^{z_1}-e^{z_2}|\leq|z_1-z_2|e^{\max\{|z_1|,|z_2|\}}.
\end{equation}
\end{lem}

\begin{proof}
First we write for every $z=x+iy\in\mathbb{C}$
\begin{equation*}
|1-e^z|^2=(1-\cos(y)e^x)^2+\sin^2(y)e^{2x}=1-2\cos(y)e^x+e^{2x}.
\end{equation*}
Using $\cos(y)\geq 1-\frac{y^2}{2}$ we can estimate this by
\begin{equation*}
|1-e^z|^2\leq(1-e^x)^2+y^2e^x.
\end{equation*}
Since $|1-e^x|\leq|x|e^{|x|}$, which follows directly from the power series representation of the exponential, the above inequality becomes
\begin{equation}\label{Eq_Exponential_estimate_1}
|1-e^z|^2\leq x^2e^{2|x|}+y^2e^x\leq(x^2+y^2)e^{2|x|}=|z|^2e^{2|\Re(z)|},\quad z\in\mathbb{C}.
\end{equation}
Next consider $z_1,z_2\in\mathbb{C}$ and assume $\Re(z_1)\leq\Re(z_2)$. Then
\eqref{Eq_Exponential_estimate_1} implies
\begin{equation*}
|e^{z_1}-e^{z_2}|=e^{\Re(z_1)}|1-e^{z_2-z_1}|\leq e^{\Re(z_1)}|z_2-z_1|e^{\Re(z_2-z_1)}=|z_2-z_1|e^{\Re(z_2)}\leq|z_2-z_1|e^{|z_2|},
\end{equation*}
and in the same way for $\Re(z_2)\leq\Re(z_1)$ one obtains $|e^{z_1}-e^{z_2}|\leq |z_1-z_2|e^{|z_1|}$.
These estimates immediately lead to \eqref{Eq_Exponential_estimate}.
\end{proof}

To verify the integral representation \eqref{Eq_Superoscillating_integral} of the function \eqref{Eq_Fdelta_sinc} in the proof of Corollary~\ref{cor_Fdelta} we need the following integral representation of the $\sinc$-function. A sketch of this proof is already given in \cite[Appendix~C]{B16}.

\begin{lem}\label{lem_Integral_representation_of_sinc}
For any $b>0$ one has
\begin{equation}\label{Eq_Integral_representation_of_sinc}
\sinc\big(\sqrt{z^2+b^2}\big)=\frac{1}{2}\int_{-1}^1e^{ikz}J_0(b\sqrt{1-k^2})dk,\quad z\in\mathbb{C}.
\end{equation}
where $J_0$ is the Bessel function of order zero.
\end{lem}

\begin{proof}
We start by deriving the so called Mehler-Sonine integral representation of the Bessel function
\begin{equation}\label{Eq_Mehler_Sonine_integral}
J_0(x)=\frac{2}{\pi}\lim\limits_{R\rightarrow\infty}\int_0^R\sin(x\cosh(t))dt,\quad x>0.
\end{equation}
Using the Cauchy theorem we can transform for every $R>0$ the integral
\begin{align*}
\int_0^Re^{ix\cosh(t)}dt&=i\int_0^{\frac{\pi}{2}}e^{ix\cosh(it)}dt+\int_0^Re^{ix\cosh(t+i\frac{\pi}{2})}dt-i\int_0^{\frac{\pi}{2}}e^{ix\cosh(R+it)}dt \\
&=i\int_0^{\frac{\pi}{2}}e^{ix\cos(t)}dt+\int_0^Re^{-x\sinh(t)}dt-i\int_0^{\frac{\pi}{2}}e^{ix\cosh(R)\cos(t)}e^{-x\sinh(R)\sin(t)}dt.
\end{align*}
Performing the limit $R\rightarrow\infty$, the last integral vanishes due to the dominated convergence theorem and we get
\begin{equation*}
\lim\limits_{R\rightarrow\infty}\int_0^Re^{ix\cosh(t)}dt=i\int_0^{\frac{\pi}{2}}e^{ix\cos(t)}dt+\int_0^\infty e^{-x\sinh(t)}dt.
\end{equation*}
Using the classical integral representation of the Bessel function \cite[Eq. (9.1.18)]{AS72} we get
\begin{align*}
J_0(x)&=\frac{1}{\pi}\int_0^\pi\cos(x\cos(t))dt=\frac{2}{\pi}\int_0^{\frac{\pi}{2}}\cos(x\cos(t))dt \\
&=\frac{2}{\pi}\Im\Big(\lim\limits_{R\rightarrow\infty}\int_0^Re^{ix\cosh(t)}dt\Big)=\frac{2}{\pi}\lim\limits_{R\rightarrow\infty}\int_0^R\sin\big(x\cosh(t)\big)dt,
\end{align*}
which is exactly the stated integral representation \eqref{Eq_Mehler_Sonine_integral}.

\medskip

For the main part of the proof let us consider the function
\begin{equation*}
S(z)\coloneqq\sinc\big(\sqrt{z^2+b^2}\big),\quad z\in\mathbb{C}.
\end{equation*}
Since the restriction $S|_\mathbb{R}$ is square integrable, its Fourier transform is given by the improper Riemann integral
\begin{equation*}
\mathcal{F}[S|_\mathbb{R}](k)=\frac{1}{\sqrt{2\pi}}\lim\limits_{R_1,R_2\rightarrow\infty}\int_{-R_1}^{R_2}e^{-ikx}S(x)dx=\frac{1}{\sqrt{2\pi}}\lim\limits_{R_1,R_2\rightarrow\infty}\int_{-R_1}^{R_2}\cos(kx)S(x)dx,\quad k\in\mathbb{R},
\end{equation*}
where the imaginary part of the integral vanishes in the limit due to the symmetry of $S$. Starting with $|k|<1$, we substitute $x=b\sinh(t)$, to write this integral as
\begin{equation}\label{Eq_Fourier_transform_integral}
\mathcal{F}[S|_\mathbb{R}](k)=\frac{b}{\sqrt{2\pi}}\lim\limits_{\widetilde{R}_1,\widetilde{R}_2\rightarrow\infty}\int_{-\widetilde{R}_1}^{\widetilde{R}_2}\cos\big(kb\sinh(t)\big)S\big(b\sinh(t)\big)\cosh(t)dt,
\end{equation}
where we used $\widetilde{R}_i\coloneqq\text{arsinh}(\frac{R_i}{b})$, $i=1,2$. Using the trigonometric identity $2\sin(u)\cos(v)=\sin(u+v)+\sin(u-v)$ we can write the integrand as
\begin{align*}
2b\cos\big(kb\sinh(t)\big)S\big(b\sinh(t)\big)\cosh(t)&=2\cos\big(kb\sinh(t)\big)\sin\big(b\cosh(t)\big) \\
&\hspace{-0.5cm}=\sin\big(b\cosh(t)+kb\sinh(t)\big)+\sin\big(b\cosh(t)-kb\sinh(t)\big).
\end{align*}
Since $|k|<1$ there exists some $t_0\in\mathbb{R}$ such that $e^{2t_0}=\frac{1+k}{1-k}$. Then, it is easy to verify that $\cosh(t_0)=\frac{1}{\sqrt{1-k^2}}$ and $\sinh(t_0)=\frac{k}{\sqrt{1-k^2}}$ and consequently
\begin{equation*}
\cosh(t)\pm k\sinh(t)=\sqrt{1-k^2}\,\cosh(t\pm t_0).
\end{equation*}
Using this identity we can further rewrite the integrand as
\begin{align*}
2b\cos\big(kb\sinh(t)\big)S\big(b\sinh(t)\big)\cosh(t)=&\sin\big(b\sqrt{1-k^2}\,\cosh(t+t_0)\big) \\
&+\sin\big(b\sqrt{1-k^2}\,\cosh(t-t_0)\big).
\end{align*}
Plugging this representation in the integral \eqref{Eq_Fourier_transform_integral} gives
\begin{align*}
\mathcal{F}[S|_\mathbb{R}](k)&=\frac{1}{2\sqrt{2\pi}}\lim\limits_{\widetilde{R}_1,\widetilde{R}_2\rightarrow\infty}\int_{-\widetilde{R}_1}^{\widetilde{R}_2}\Big(\sin\big(b\sqrt{1-k^2}\,\cosh(t+t_0)\big) \\
&\hspace{5cm}+\sin\big(b\sqrt{1-k^2}\,\cosh(t-t_0)\big)\Big)dt \\
&=\frac{1}{\sqrt{2\pi}}\lim\limits_{\widehat{R}_1,\widehat{R}_2\rightarrow\infty}\int_{-\widehat{R}_1}^{\widehat{R}_2}\sin\big(b\sqrt{1-k^2}\,\cosh(t)\big)dt \\
&=\frac{\sqrt{2}}{\sqrt{\pi}}\lim\limits_{\widehat{R}_2\rightarrow\infty}\int_0^{\widehat{R}_2}\sin\big(b\sqrt{1-k^2}\,\cosh(t)\big)dt \\
&=\frac{\sqrt{\pi}}{\sqrt{2}}J_0\big(b\sqrt{1-k^2}\big),
\end{align*}
where in the second line we substituted $t\rightarrow t\mp t_0$, used $\widehat{R}_1\coloneqq\widetilde{R}_1\mp t_0$, $\widehat{R}_2\coloneqq\widetilde{R}_2\pm t_0$, in the two respective integrals and added them together. Moreover, in the last equation we used the representation \eqref{Eq_Mehler_Sonine_integral}.

\medskip

For $k>1$ we use the Cauchy theorem to change the integration path to a semicircle in the lower half space
\begin{equation}\label{Eq_Fourier_F}
\begin{split}
\mathcal{F}[S|_\mathbb{R}](k)&=\frac{1}{\sqrt{2\pi}}\lim\limits_{R_1,R_2\rightarrow\infty}\int_{-R_1}^{R_2}e^{-ikx}S(x)dx\\
&=\frac{1}{\sqrt{2\pi}}\lim\limits_{R\rightarrow\infty}\int_{-R}^{R}e^{-ikx}S(x)dx\\
&=\frac{i}{\sqrt{2\pi}}\lim\limits_{R\rightarrow\infty}\int_\pi^{2\pi}R\,e^{-ikRe^{i\varphi}}S(Re^{i\varphi})e^{i\varphi}d\varphi.
\end{split}
\end{equation}
Choosing $R>b$ and using that
\begin{equation}\label{Eq_Radius_estimate}
R^2-b^2\leq|R^2e^{2i\varphi}+b^2|\leq R^2+b^2,\quad\varphi\in[\pi,2\pi],
\end{equation}
the function $S$ in the integrand can be estimated by
\begin{equation}\label{Eq_S_estimate}
|S(Re^{i\varphi})|=\bigg|\frac{\sin(\sqrt{R^2e^{2i\varphi}+b^2})}{\sqrt{R^2e^{2i\varphi}+b^2}}\bigg|\leq\frac{e^{|\Im(\sqrt{R^2e^{2i\varphi}+b^2})|}}{\sqrt{R^2-b^2}},\quad\varphi\in[\pi,2\pi].
\end{equation}
Writing any $w\in\mathbb{C}$ as $w=|w|e^{i\Arg(w)}$ with $\Arg(w)\in[0,2\pi)$, our choice \eqref{Eq_Root} of the complex square root shows $\sqrt{w}=\sqrt{|w|}\,e^{i\frac{\Arg(w)}{2}}$. Using this, we  rewrite the imaginary part
\begin{equation*}
\Im(\sqrt{w})=\sqrt{|w|}\,\sin\Big(\frac{\Arg(w)}{2}\Big)=\frac{\sqrt{|w|}}{\sqrt{2}}\sqrt{1-\cos(\Arg(w))}=\frac{1}{\sqrt{2}}\sqrt{|w|-\Re(w)},\quad w\in\mathbb{C}.
\end{equation*}
This representation and the estimate \eqref{Eq_Radius_estimate} can now be used to estimate the exponent in \eqref{Eq_S_estimate} by
\begin{equation*}
2\big|\Im\big(\sqrt{R^2e^{2i\varphi}+b^2}\big)\big|^2=|R^2e^{2i\varphi}+b^2|-\Re(R^2e^{2i\varphi}+b^2)\leq R^2-R^2\cos(2\varphi)=2R^2\sin^2(\varphi).
\end{equation*}
Hence we can further estimate \eqref{Eq_S_estimate} by
\begin{equation*}
|S(Re^{i\varphi})|\leq\frac{e^{-R\sin(\varphi)}}{\sqrt{R^2-b^2}},\quad\varphi\in[\pi,2\pi].
\end{equation*}
Using this inequality, the Fourier transform \eqref{Eq_Fourier_F} can be estimated by
\begin{align*}
|\mathcal{F}[S|_\mathbb{R}](k)|&\leq\frac{1}{\sqrt{2\pi}}\lim\limits_{R\rightarrow\infty}\int_\pi^{2\pi}\frac{R}{\sqrt{R^2-b^2}}e^{(k-1)R\sin(\varphi)}d\varphi \\
&=\frac{1}{\sqrt{2\pi}}\int_\pi^{2\pi}\lim\limits_{R\rightarrow\infty}\frac{R}{\sqrt{R^2-b^2}}e^{(k-1)R\sin(\varphi)}d\varphi\\
&=0.
\end{align*}
For $k<-1$ we also conclude $\mathcal{F}[S|_\mathbb{R}](k)=0$ by the symmetry of $S$.

\medskip

Summing up, we have now shown that
\begin{equation}\label{Eq_Fourier_S}
\mathcal{F}[S|_\mathbb{R}](k)=\begin{cases} \frac{\sqrt{\pi}}{\sqrt{2}}J_0(b\sqrt{1-k^2}), & \text{if }|k|<1, \\ 0, & \text{if }|k|>1. \end{cases}
\end{equation}
Applying the inverse Fourier transform to \eqref{Eq_Fourier_S} we obtain
\begin{equation*}
S(x)=\frac{1}{2}\int_{-1}^1e^{ikx}J_0(b\sqrt{1-k^2})dk,\quad x\in\mathbb{R}.
\end{equation*}
Since $J_0$ is a bounded function, the right hand side extends to an entire function when $x\in\mathbb{R}$ is replaced by $z\in\mathbb{C}$. Since the holomorphic extension is unique it coincides with $S$ and we conclude the stated integral representation \eqref{Eq_Integral_representation_of_sinc}.
\end{proof}

\end{document}